\documentclass[journal]{IEEEtran}

\usepackage{xcolor,soul,framed} 

\colorlet{shadecolor}{yellow}
\usepackage[pdftex]{graphicx}
\graphicspath{{../pdf/}{../jpeg/}}
\DeclareGraphicsExtensions{.pdf,.jpeg,.png}

\usepackage{cite}
\usepackage{amsmath,amssymb,amsfonts}
\usepackage{textcomp}
\usepackage{xcolor}
\usepackage{booktabs}
\usepackage{multirow}
\usepackage{subfigure}
\usepackage{bm}
\usepackage{algorithmic} 
\usepackage{algorithm}
\usepackage{array}
\usepackage{mdwmath}
\usepackage{mdwtab}
\usepackage{eqparbox}
\usepackage{url}
\usepackage{bbm}
\usepackage{mathabx}

\DeclareMathOperator*{\argmin}{arg\,min}
\newtheorem{theorem}{Theorem}

\newtheorem{definition}{Definition}
\newtheorem{remark}{Remark}

\begin{document}
\bstctlcite{IEEEexample:BSTcontrol}
    \title{MVPatch: More Vivid Patch for Adversarial Camouflaged Attacks on Object Detectors in the Physical World}
  \author{Zheng~Zhou,~\IEEEmembership{Graduate Student Member,~IEEE,} 
      Hongbo~Zhao,~\IEEEmembership{Senior Member,~IEEE,} \\
	  Ju~Liu,~\IEEEmembership{Senior Member,~IEEE,} 
      Qiaosheng~Zhang, 
	  Liwei~Geng, \\Shuchang~Lyu,~\IEEEmembership{Graduate Student Member,~IEEE,} 
      and Wenquan Feng

\thanks{This work was supported by the National Natural Science Foundation of China under Grant 61901015. (\textit{Corresponding author: Hongbo Zhao.})}
  \thanks{Zheng Zhou, Hongbo Zhao, Liwei Geng, Shuchang Lyu and Wenquan Feng are with the School of Electronic and Information Engineering, Beihang University, Beijing, 100191, China and Wenquan Feng is also with Hefei Inovation Research Institute of Beihang University, Beihang University, Hefei, 230071, China (e-mail: zhengzhou@buaa.edu.cn; bhzhb@buaa.edu.cn; liviageng@buaa.edu.cn; lyushuchang@buaa.edu.cn; buaafwq@buaa.edu.cn).}
  \thanks{Ju Liu is with the School of Information Science and Engineering, Shandong University, Qingdao, 266237, China (e-mail: juliu@sdu.edu.cn).}%
  \thanks{Qiaosheng Zhang is with Shanghai Artificial Intelligence Laboratory, Shanghai, 200032, China (e-mail: zhangqiaosheng@pjlab.org.cn).}
}


\maketitle

\begin{abstract}
Recent studies have shown that Adversarial Patches (APs) can effectively manipulate object detection models. However, the conspicuous patterns often associated with these patches tend to attract human attention, posing a significant challenge. Existing research has primarily focused on enhancing attack efficacy in the physical domain while often neglecting the optimization of stealthiness and transferability. Furthermore, applying APs in real-world scenarios faces major challenges related to transferability, stealthiness, and practicality. To address these challenges, we introduce generalization theory into the context of APs, enabling our iterative process to simultaneously enhance transferability and refine visual correlation with realistic images. We propose a Dual-Perception-Based Framework (DPBF) to generate the \textit{More Vivid Patch} (\textit{MVPatch}), which enhances transferability, stealthiness, and practicality. The DPBF integrates two key components: the Model-Perception-Based Module (MPBM) and the Human-Perception-Based Module (HPBM), along with regularization terms. The MPBM employs ensemble strategy to reduce object confidence scores across multiple detectors, thereby improving AP transferability with robust theoretical support. Concurrently, the HPBM introduces a lightweight method for achieving visual similarity, creating natural and inconspicuous adversarial patches without relying on additional generative models. The regularization terms further enhance the practicality of the generated APs in the physical domain. Additionally, we introduce naturalness and transferability scores to provide an unbiased assessment of APs. Extensive experimental validation demonstrates that MVPatch achieves superior transferability and a natural appearance in both digital and physical domains, underscoring its effectiveness and stealthiness.
\end{abstract}

\begin{IEEEkeywords}
Adversarial example, patch attack, physical attack, neural network, transferable and stealthy attack
\end{IEEEkeywords}

%
\IEEEpeerreviewmaketitle


\section{Introduction}
\begin{figure}[htbp]
	\centering
	\includegraphics[width=1\linewidth]{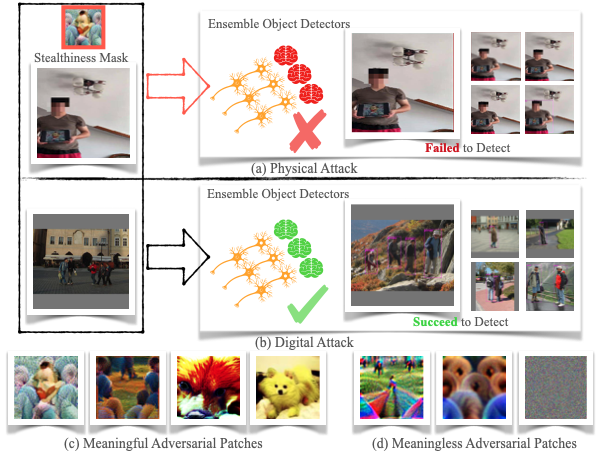}
	\caption{Introduction to the attack scenarios of MVPatch and illustration of the meaningful and meaningless adversarial patches. (a) demonstrates that MVPatch can make a person invisible to detectors in the real world, while (b) demonstrates that detectors can successfully distinguish a person in the digital world. (c) and (d) demonstrate both diverse meaningful adversarial patches \cite{b96} and meaningless adversarial patches \cite{b87,b100}, respectively.}
	\label{f12}
\end{figure}
\IEEEPARstart{D}{eep} Neural Networks (DNNs) have achieved remarkable performance in manifold fields, such as computer vision \cite{b133,b134,b135,b137}, natural language processing\cite{b138,b139}, and automatic speech recognition \cite{b140,b141}. However, due to their lack of interpretability, DNNs are vulnerable to Adversarial Examples (AEs) \cite{b1,b2,b136,b161}, which raises concerns about their reliability in security-critical applications such as face recognition \cite{b92,b119} and autonomous driving \cite{b122,b123}. In general, AE attacks can be categorized into two types: Digital Attack (DA) \cite{b1,b2,b3,b4,b5,b32,b8,b44,b62,b113}, which involves introducing digital perturbations into the input image to carry out the attack, and Physical Attack (PA) \cite{b114,b111,b109,b115,b93,b87,b99,b96,b101,b116,b100,b108,b119,b121,b107,b124}, which directly targets real-world objects using adversarial perturbations, as illustrated in Fig. \ref{f12}(b) and Fig. \ref{f12}(a), respectively. 

Object Detection Models (ODMs) \cite{b142,b143,b144}, comprised of deep neural layers, are extensively utilized in real-world applications such as people tracking \cite{b145}, pedestrian re-identification (Re-ID) \cite{b146}, and remote sensing \cite{b148}. Although ODMs enhance human convenience and improve the quality of life, they simultaneously pose a considerable risk to individual privacy and sensitive personal information \cite{b96,b100,b87,b99}. Furthermore, the capability of physical attacks to be transferred amplifies the vulnerability of ODMs \cite{b130,b131,b132,b164,b165,b166}. In real-world scenarios, AEs could be often observed which are quasi-imperceptible to humans. For example, as depicted in Fig. \ref{f12}, an intruder appears under a surveillance camera holding an adversarial patch, which is camouflaged as Van Gogh's painting, and the DNN-based object detection system will fail. It can even be said that the intruder is invisible under the camera, which will seriously threaten the public safety system  Hence, it is imperative to conduct a thorough investigation into PA on ODMs. 

Adversarial Patches (APs) \cite{b115} have emerged as a potent means to manipulate ODMs within the physical realm, offering several advantages such as input-independence and scene-independence, with substantial real-world impacts \cite{b32,b92,b94,b98,b100,b101}. However, despite their potential, adversarial patches face three major issues that hinder their broader applicability and effectiveness: \textit{Transferability}, \textit{Stealthiness}, and \textit{Practicality}. Transferability can be enhanced through methods like gradient-based optimization (e.g., MI \cite{b149}, NI \cite{b150}), input transformation (e.g., DI \cite{b151}, SI \cite{b150}), and model ensemble methods (e.g., MI \cite{b149}, SVRE \cite{b152}), with the latter being the most effective yet not fully understood. Stealthiness involves generating realistic APs using GAN-based methods (e.g., Natural Patch \cite{b96}), Diffusion-based methods (e.g., Diff-PGD \cite{b107}), and Regularization-based methods (e.g., AdvART \cite{b124}, DAP \cite{b153}), each with trade-offs in computational cost and efficiency. Practicality concerns the transition from digital to physical domains \cite{b109,b111,b114,b116}, where APs face challenges like performance degradation due to varying distances and angles. 
Moreover, much of the existing literature either focuses on enhancing the transferability of adversarial patches \cite{b130,b131,b132,b162} while overlooking their stealthiness, or prioritizes stealthiness \cite{b87,b93,b95,b96,b99,b100,b108,b109,b110,b111,b124} at the expense of transferability. Notably, there is limited research that simultaneously addresses both stealthiness and transferability, while also considering computational resource efficiency.

To address these challenges, we propose a Dual-Perception-Based Framework (DPBF) for generating a powerful adversarial patch, termed the \textbf{M}ore \textbf{V}ivid \textbf{Patch} (\textbf{MVPatch}). This framework enhances the transferability, stealthiness, and practicality of adversarial patches for attacking object detection models (ODMs). To improve transferability, we introduce the Model-Perception-Based Module (MPBM), which employs an ensemble strategy with multiple object detectors. The MPBM leverages the ensemble strategy to boost adversarial attack transferability, supported by theoretical insights from generalization theory. For enhanced stealthiness, we present the Human-Perception-Based Module (HPBM). This module generates adversarial patches that closely mimic a specified image in Hilbert space, using diverse transformations and the Frobenius norm to minimize generalization error. This approach ensures the patches are natural and inconspicuous without relying on additional generative models. To improve real-world practicality, we incorporate regularization terms such as Total Variation (TV) and Non-Printable Score (NPS) into the DPBF. By integrating the MPBM, HPBM, and regularization terms, we create a more effective adversarial patch with significant transferability and stealthiness while maintaining practicality. Additionally, we introduce naturalness and transferability scores as experimental metrics to assess the performance of the adversarial patches. We evaluate these patches using threat models including YOLOv2, YOLOv3, YOLOv3-tiny, YOLOv4, and YOLOv4-tiny, as well as transfer attack models such as Faster R-CNN, SSD, and YOLOv5. Extensive experiments in both digital and physical scenarios, along with independent subjective surveys, demonstrate that our MVPatch consistently outperforms several state-of-the-art baselines in terms of naturalness and transferability. \\
The \textbf{main contributions} can be summarized as follows:
\begin{itemize}
\item To the best of our knowledge, we are the \textit{first} to introduce generalization theory to the field of adversarial patches. This approach provides robust theoretical support for generating patches that maintain transferability, stealthiness, and practicality through our Dual-Perception-Based Framework (DPBF).
\item We present the Model-Perception-Based Module (MPBM), which reduces object confidence scores through an ensemble strategy. Utilizing generalization theory, the MPBM demonstrates that this ensemble approach enhances both generalization and stability, thereby improving the transferability of adversarial patches (APs) compared to single-model strategies.
\item We propose a Human-Perception-Based Module (HPBM) that generates natural and stealthy adversarial patches without the need for additional generative models, while ensuring minimal generalization error.
\item We employ naturalness and transferability scores as evaluation metrics to assess the naturalness and attack transferability of diverse adversarial patches in our experiments.
\item We perform a comprehensive analysis of the transfer attack performance and naturalness of the proposed method, using both meaningful and meaningless patch approaches in a variety of digital and physical scenarios.
\end{itemize}

The remainder of this paper is organized as follows: Section \ref{s2} reviews the literature on adversarial examples, covering both digital and physical domains, and explores generalization aspects for adversarial attacks. Section \ref{s3} defines adversarial patch attacks, discusses their characteristics, and details the input and output of the threat models, as well as the generation process of the proposed MVPatch. Section \ref{s4} describes the experimental environment and dataset, including the experimental setup and evaluation metrics. Section \ref{s5} demonstrates the empirical effectiveness of MVPatch through extensive evaluations in both digital and physical settings. Finally, Section \ref{s6} summarizes the findings, discusses limitations, and outlines future research directions.

\section{Related Work}\label{s2}
\subsection{Digital Attack}
Adversarial Examples (AEs) are perturbations that are imperceptible to humans but can mislead Deep Neural Networks (DNNs). In 2014, Szegedy {\it et al.} \cite{b1} first discovered AEs, successfully attacking a DNN using the L-BFGS algorithm. In 2015, Goodfellow {\it et al.} \cite{b2} introduced the Fast Gradient Sign Method (FGSM), enhancing attack effectiveness on neural network classifiers. Kurakin {\it et al.} \cite{b32} developed the Basic Iterative Method (BIM) and demonstrated AEs' applicability in the physical world. Papernot {\it et al.} \cite{b8} presented the Jocobian-Based Saliency Map Attack (JSMA) for partial attacks. Madry {\it et al.} \cite{b3} proposed Project Gradient Descent (PGD) to enhance DNN robustness through a combined attack and defense framework. Other notable digital adversarial attacks include Deepfool \cite{b5}, C\&W \cite{b4}, ZOO \cite{b44}, Universal Perturbation \cite{b62}, One-Pixel attack \cite{b113}, and attacks on hasing retrival \cite{b163} and large language vision models \cite{b127,b128,b129}.
\subsection{Physical Attack}
Recently, it has been observed that printed Adversarial Examples (AEs) can effectively deceive neural network models in the physical domain. Athalya {\it et al.} \cite{b114} developed robust 3D AEs mitigating diverse angle transformations. Sharif {\it et al.} \cite{b111} and Komkov {\it et al.} \cite{b109} attacked facial recognition systems with adversarial eyeglasses and hats, respectively. Brown {\it et al.} \cite{b115} introduced the concept of Adversarial Patches (APs). Liu {\it et al.} \cite{b93} proposed DPatch for object detection models, and Thy {\it et al.} \cite{b87} introduced AdvPatch for misleading surveillance systems. Wu {\it et al.} \cite{b99} presented AdvCloak to make humans invisible to object detectors. Hu {\it et al.} \cite{b96} designed a more natural patch based on AdvPatch, and Xu {\it et al.} \cite{b101} created an adversarial T-shirt for object detection. Huang {\it et al.} \cite{b116} introduced T-SEA for ensemble attacks on object detectors. Hu {\it et al.} \cite{b100} proposed TC-EGA, applying adversarial textures to clothing for stealthiness. Wang {\it et al.} \cite{b108} developed an adversarial patch for car camouflage using attention mechanisms. Wei {\it et al.} \cite{b119} focused on enhancing attack performance via spatial positioning of patches. Liu {\it et al.} \cite{b121} proposed PS-GAN for improved visual fidelity and attack performance. Xue {\it et al.} \cite{b107} introduced DIff-PGD to generate realistic AEs using diffusion models. Guesmi {\it et al.} \cite{b124} proposed AdvART for attacking object detectors without generative models.
\subsection{Generalization Analysis for Adversarial Patch Attack}
Generalization evaluates how well a model trained on a dataset performs on unseen data, typically using Empirical Risk Minimization (ERM) \cite{b155}. Given data $\{(x_i,y_i),i=1,\dots,n\}$ from a probability distribution $\mathcal{D}$ on $\Omega\times \{-1,1\}$, a class of functions $\mathcal{F}:\Omega\rightarrow\mathbb{R}$ and and a loss function $l$, ERM finds a minimizer of the empirical loss:
\begin{equation}
	f^* = \argmin_{f\in\mathcal{F}} \mathcal{L}_{\text{emp}}(f):= \argmin_{f\in\mathcal{F}}\sum_{i}l(f(x_i),y_i).
\end{equation}
The complexity of the model $\mathcal{F}$ can be assessed using approaches such as VC-dimension \cite{b156}, Rademacher complexity \cite{b157}, and covering numbers \cite{b158}. These analyses typically provide bounds on the generalization gap, which is the difference between the empirical loss and the expected loss:
\begin{equation}
	\label{e2}
	\vert\left[l(f^*(x),y)\right]-\mathcal{L}_{\text{emp}}(f^*(x))\vert < O^*(\sqrt{\frac{c}{n}}).
\end{equation}
From Eq. \ref{e2}, we observe that the generalization error depends on both model complexity $c$ and the number of samples $n$. Recent studies have investigated the relationship between adversarial robustness and generalization \cite{b159, b160}. Huang {\it et al.} \cite{b116} introduced T-SEA, which utilizes the insights from Eq. \ref{e2} to improve the transferability of generated adversarial patches.

Despite significant advancements in adversarial patches for the physical world, challenges persist, particularly in balancing transferability \cite{b116,b96,b120} and stealthiness \cite{b114,b109,b99,b101,b100,b108,b118,b121,b107,b124}. Recent research often focuses on generating meaningless adversarial patches that lack visual coherence for human perception, as shown in Fig. 1(d). Hu {\it et al.} \cite{b96} systematically explored meaningful and transferable adversarial patches using GANs, although training GANs is computationally expensive. Moreover, there is a lack of robust theoretical support for addressing the issues of transferability, stealthiness, and practicality. To address these challenges, we propose a Dual-Perception-Based Framework (DPBF) to generate the \textbf{M}ore \textbf{V}ivid \textbf{Patch} (\textbf{MVPatch}) with enhanced transferability, stealthiness, and practicality. Inspired by T-SEA \cite{b116}, we incorporate generalization theory to comprehensively analyze and improve the generalization and stability of adversarial patches, enhancing both their transferability and stealthiness. Our MVPatch outperforms state-of-the-art methods in both digital and physical scenarios, achieving superior transferability and stealthiness with a lightweight approach.
\section{Methodology}\label{s3}
\begin{figure*}[htbp]
	\centering
	\includegraphics[width=1\linewidth]{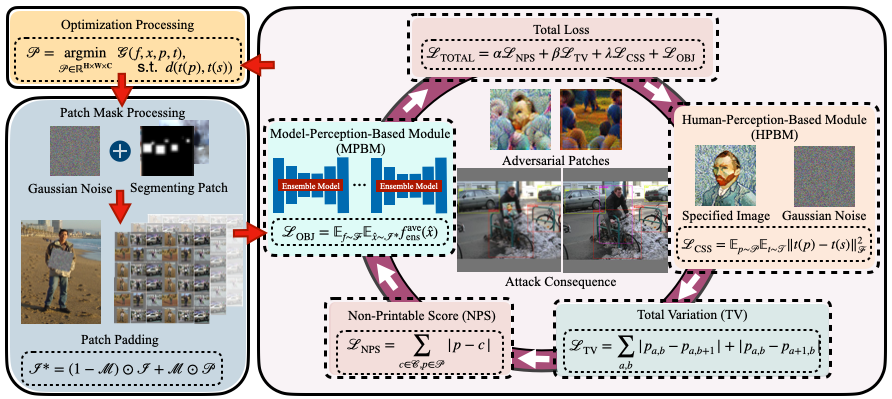}
	\caption{The MVPatch pipeline involves embedding masks into benign images and applying them to object detectors to determine the object confidence scores. To achieve improved attack performance, the algorithm minimizes various losses, such as TV loss, NPS loss, OBJ loss, and CSS loss, to obtain the optimal adversarial patch. Algorithm \ref{a1} illustrates the complete procedure for the MVPatch algorithm.}
	\label{f3}
\end{figure*}
In this section, we provide a comprehensive overview of our Dual-Perception-Based Framework (DPBF) for generating \textbf{M}ore \textbf{V}ivid \textbf{Patch} (\textbf{MVPatch}) applicable to both digital and physical domains. We start by outlining the motivation behind the development of MVPatch. Following this, we introduce the foundational concepts, including notations, definitions, and problem formulation. We then detail the specifics of our proposed adversarial patch attack framework. Finally, we summarize the entire MVPatch generation process.
\subsection{Motivation}
Adversarial Patches (APs) have gained significant attention in adversarial machine learning for their ability to manipulate object detection models. These patches can be strategically placed to cause misclassifications, posing challenges for robust detection systems. However, APs face three major challenges: \textit{Transferability}, \textit{Stealthiness}, and \textit{Practicality}.

\begin{itemize} 
	\item \textit{Transferability}: There is insufficient theoretical support explaining why ensemble methods improve transferability, indicating a need for further research. 
	\item \textit{Stealthiness}: GAN-based and diffusion-based methods generate realistic APs but at high computational costs and longer processing times, presenting a trade-off between realism and efficiency. 
	\item \textit{Practicality}: In the physical domain, AP performance degrades with varying distances and angles, necessitating solutions for effective real-world deployment. 
\end{itemize}

To address these challenges, we propose a Dual-Perception-Based Framework (DPBF) designed to generate effective and stealthy adversarial patches, termed \textbf{MVPatch}. This framework improves the transferability of adversarial patches across various models while also enhancing their visual inconspicuousness and robustness in physical environments. We frame the adversarial transferability problem as an optimal generalization issue and introduce two modules to minimize generalization error: the Model-Perception-Based Module (MPBM), which enhances transferability through ensemble strategy, and the Human-Perception-Based Module (HPBM), which improves patch stealthiness using the Frobenius norm in Hilbert space transformations. Additionally, we incorporate regularization terms such as Total Variation (TV) and Non-Printable Score (NPS) into the DPBF to improve practical deployment in real-world scenarios. The overall framework is illustrated in Figure \ref{f3}.
\subsection{Preliminary}
\noindent \textbf{Notations.} We consider neural networks $\mathcal{F}(\mathcal{I}, \mathcal{Y})$ for object detection tasks and let $\mathcal{I}$ be the input space of the model, and let $\mathcal{Y} = \{\mathrm{a},\mathrm{b},\mathrm{w},\mathrm{h}, \mathrm{obj},\mathrm{cls}\}$. Here, $\mathrm{a}$ and $\mathrm{b}$ represent the horizontal and vertical coordinates of the detection frame, respectively. $\mathrm{w}$ and $\mathrm{h}$ denote the width and height of the detection frame. $\mathrm{obj}$ and $\mathrm{cls}$ are the object confidence scores and classification scores. We model the neural network as a mapping function $\mathcal{F}:\mathcal{I}\subseteq\mathbb{R}^\mathbf{H \times W \times C} \rightarrow \mathcal{Y}\subseteq\mathbb{R}^N$. The notations $f \in \mathcal{F}$, $x\in\mathcal{I}$, and $y\in\mathcal{Y}$ are the realizations of random variables. 

For a given input $x \in \mathcal{I}$, the object detection model $f \in \mathcal{F}$ first predicts the confidence score for each label $y_{\mathrm{obj}} \in \mathcal{Y}$, denoted as $f_{\mathrm{obj}}(x)$. These confidence scores sum up to 1, i.e., $\sum_{y_{\mathrm{obj}}\in\mathcal{Y}}f_{\mathrm{obj}}(x)=1, \forall x \in \mathcal{I}$.

For a given model $f\in\mathcal{F}$, there is usually a model-dependent loss function $\mathcal{L}_{f\in\mathcal{F}}:x\in\mathcal{I}\times y\in\mathcal{Y}\rightarrow \mathbb{R}^+$, which is the composition of a differentiable training loss, (e.g., cross-entropy loss) $\mathcal{L}$ and model's confidence score $f_{\mathrm{obj}}(\cdot)$. Thus, $\mathcal{L}(x,y):= \mathcal{L}(f(x),y)$, $(x,y)\in(\mathcal{I},\mathcal{Y})$.
\begin{definition}[Adversarial Patch] 
	An adversarial patch is a perturbation $\mathcal{P}$ applied to a benign image $\mathcal{I}$ to produce a malicious image $\mathcal{I}^*$, with the intent of deceiving object detection models $f \in \mathcal{F}$. Formally, it can be expressed as: 
	\begin{equation} 
		f(\mathcal{I}^) \neq f(\mathcal{I}) = \mathcal{Y}, \quad \text{where} \quad \mathcal{I}^* = (1 - \mathcal{M}) \odot \mathcal{I} + \mathcal{M} \odot \mathcal{P}, 
	\end{equation} 
	where $\mathcal{M}$ denotes the patch mask that determines the size, shape, and placement of the adversarial patch $\mathcal{P}$. The symbol $\odot$ represents the Hadamard product, which involves element-wise multiplication of the corresponding entries in the matrices $\mathcal{M}$, $\mathcal{I}$, and $\mathcal{P}$. 
\end{definition}
\begin{definition}[Generalization Error] 
	\label{def2} 
	Consider a training dataset $\mathcal{T} = \{(x_1, y_1), (x_2, y_2), \ldots, (x_N, y_N)\}$ drawn from a joint probability distribution $P(\mathcal{I}, \mathcal{Y})$. Let the hypothesis space consist of a set of models $\mathcal{F} = \{f_1, f_2, \ldots, f_M\}$, where $M$ denotes the number of models. The generalization error $\mathcal{G}(f)$ for a model $f \in \mathcal{F}$ is defined as: 
	\begin{equation} 
		\mathcal{G}(f) \triangleq \underbrace{\mathbb{E}\left[\mathbbm{1}\{f(x) \neq y\}\right]}_{\text{Expected Error}} - \underbrace{\frac{1}{N}\sum_{i=1}^{N}\mathbbm{1}\{f(x_i) \neq y_i\}}_{\text{Empirical Error}}, 
	\end{equation} 
	where $\mathbbm{1}\{\cdot\}$ is the indicator function that evaluates to 1 if the condition is true and 0 otherwise. 
\end{definition}
\textbf{Problem Formulation.} 
To generate adversarial patches that are both transferable and stealthy, we introduce the generalization error $\mathcal{G}(\cdot)$ and formulate the problem as an optimization task. The objective is to find an image patch $p \in \mathcal{P}$ of dimensions $(H, W, C)$, along with transformations $t \in T$, that minimize the generalization error across multiple neural networks $f \in \mathcal{F}$: 
\begin{equation} 
	\mathcal{P} = \argmin_{\mathcal{P} \in \mathbb{R}^{H \times W \times C}} \mathcal{G}(f, x, p, t), \quad \text{s.t.} \quad d(t(p), t(s)), 
\end{equation} 
where $d(\cdot)$ is the distance function that quantifies the similarity between the adversarial patch $p \in \mathcal{P}$ and a specified image $s \in \mathcal{S}$. Minimizing $\mathcal{G}(f, x, p, t)$ ensures that the generated adversarial patch maintains high transferability, while satisfying $d(t(p), t(s))$ guarantees its high stealthiness.
\subsection{Model-Perception-Based Module} 
\label{MPBM}
To enhance the \textit{transferability} of Adversarial Patches (APs), we propose the Model-Perception-Based Module (MPBM). The MPBM incorporates generalization error analysis to derive the generalization error bounds for our method. By comparing these bounds with those of single-model and ensemble-model strategies, we divide the ensemble approach into two distinct objective functions to validate the efficacy of our method. This enables us to substantiate our approach through both theoretical derivation and empirical verification. We begin with theoretical analyses focused on generalization and stability, identifying key factors that contribute to improved transferability across various models.
\begin{definition}[Ensemble Model]
    \label{def3}
    Given a set of models $\{f_i\}^M_{i=1}$, ensemble strategies can be divided into averaging outputs of neural networks and maximizing outputs of neural networks. By ensembling these models, we obtain a new predictive model:
    \begin{equation}
        f_{\text{ens}}(x) = 
        \begin{cases} 
            \frac{1}{M}\sum_{i=1}^{M}f_i(x) & f_{\text{ens}}^{\text{ave}}\,\,\,\text{for averaging outputs}, \\
            \max \{f_i(x)\}^{M}_{i=1} & f_{\text{ens}}^{\text{max}}\,\,\,\text{for maximizing outputs}.
        \end{cases}
    \end{equation}
\end{definition}
\begin{remark}
	Here, $f_{\text{ens}}^{\text{max}}$ is a special case where the ensemble model strategy can also be considered a separate model. Therefore, we analyze it as an independent model in a broad sense. All subsequent references to $f_{\text{ens}}$ refer specifically to $f_{\text{ens}}^{\text{ave}}$.
\end{remark}
\subsubsection{Generalization Analysis}
\begin{theorem}[Generalization Error of Ensemble Model]
    \label{th1}
    Given an ensemble model $f_{\text{ens}}$ and a single model $f$ from the model set $\mathcal{F}$, where $(x, y)$ are drawn from the distribution $\mathcal{D}$, with $x$ and $y$ representing the input and output spaces of the models, respectively. Then,
    \begin{align}
        \mathbb{E}_{(x,y)\sim \mathcal{D}; f\sim \mathcal{F}}\left[\mathcal{L}(f_{\text{ens}}(x), y)\right] 
        \leq \mathbb{E}_{(x,y)\sim \mathcal{D}; f\sim \mathcal{F}}\left[\mathcal{L}(f(x), y)\right].
    \end{align}
\end{theorem}

\begin{proof}
	\begin{align}
		&\mathbb{E}_{(x,y)\sim \mathcal{D}}\left[\mathcal{L}(f_{\text{ens}}(x),y)\right] 
		= \mathbb{E}_{(x,y)\sim \mathcal{D}}\left[\mathcal{L}\left(\frac{1}{M}\sum_{i=1}^{M}f_i(x),y\right)\right] \\
		&\leq \frac{1}{M}\sum_{i=1}^{M} \mathbb{E}_{(x,y)\sim \mathcal{D}}\left[\mathcal{L}(f_i(x),y)\right] \quad \text{(by Jensen's inequality)} \\
		&= \mathbb{E}_{(x,y)\sim \mathcal{D};f\sim\mathcal{F}}\left[\mathcal{L}(f(x),y)\right].
	\end{align}
\end{proof}
\begin{remark}
	Since the loss function $\mathcal{L}$ is typically convex, we can utilize Jensen's inequality to establish the relationship between the ensemble model $f_{\text{ens}}$ and single model $f$ as shown in Theorem \ref{th1}. Specifically, Jensen's inequality ensures that the loss of the ensemble model, which is a convex combination of the individual models, will not exceed the average loss of the individual models. This highlights that ensemble model generally exhibit better generalization error compared to individual models.
\end{remark}
\subsubsection{Stability Analysis}
\begin{theorem}[Stability of the Ensemble Model]
	Let $\text{Var}$ denote variance and $\text{Cov}$ denote covariance. Suppose the individual neural networks $f_i, f_j \in \mathcal{F}$ each have a variance of $\sigma^2$. The variance of the ensemble model $f_{\text{ens}}$ can be expressed as: 
	\begin{equation}
		\label{e10}
		\text{Var}(f_{\text{ens}}) = \frac{\sigma^2}{M}\left( 1+2\rho\right),
	\end{equation}
	where $\rho = \frac{\text{Cov}(f_i(x), f_j(x))}{\sqrt{\text{Var}(f_i(x)) \text{Var}(f_j(x))}}$ is the correlation coefficient between pairs of models, and $M$ is the number of models.
\end{theorem}
\begin{proof}
	\begin{align}
		\text{Var}(f_{\text{ens}}) &= \text{Var}\left[\frac{1}{M}\sum_{i=1}^{M}f_i(x)\right]\\
		&= \frac{1}{M^2}\text{Var}\left[\sum_{i=1}^{M}f_i(x)\right]\\
		&= \frac{1}{M^2}\underbrace{\left[\sum_{i=1}^{M}\text{Var}\left(f_i(x)\right) + 2\sum_{i \neq j}\text{Cov}\left(f_i(x),f_j(x)\right)\right]}_{\text{Bienaymé's identity}}\\
		&= \frac{1}{M^2}\left(M\sigma^2 + 2M\rho\sigma^2\right)\\
		&= \frac{\sigma^2}{M}\left( 1+2\rho\right).
	\end{align}
\end{proof}
\begin{remark}
	The variance of the ensemble model $f_{\text{ens}}$ is proportional to the correlation coefficient $rho$ and inversely proportional to the number of models $M$. Consequently, Eq. \ref{e10} suggests that as the number of models $M$ increases, the variance will decrease, provided that the correlation between the models $\rho$ remains low. This reduction in variance implies an enhancement in both the stability and the generalization capability of the ensemble model. Thus, the adversarial patch generation process $\hat{x}\in\mathcal{I^*}$ using the ensemble model $f_{\text{ens}}$ can achieve greater stability compared to a single model approach, including strategy that aims to maximize outputs in an ensemble model.
\end{remark}

Therefore, we adopt the ensemble model strategy as outlined in Definition \ref{def3}. The corresponding objective function is formulated as:
\begin{equation}
	\mathcal{L}_{\text{OBJ}} = \mathbb{E}_{f\sim\mathcal{F}}\mathbb{E}_{\hat{x}\sim\mathcal{I^*}}f_{\text{ens}}^{\text{ave}}(\hat{x}).
\end{equation}
The MPBM leverages the averaged outputs of the ensemble model $f_{\text{ens}}^{\text{ave}}$, ensuring enhanced stability and generalization during the adversarial patch generation process. By incorporating multiple models with low correlation, we achieve a more robust and reliable performance, leading to the conclusion that the ensemble strategy effectively mitigates variance and improves overall model efficacy.
\subsection{Human-Perception-Based Module}
To enhance the \textit{stealthiness} of Adversarial Patches (APs), we present the Human-Perception-Based Module (HPBM). This approach employs Kullback-Leibler (KL) divergence as a distance function and introduces a novel metric, the Compared Specified Image Similarity (CSS) measurement, to assess the similarity between the generated patch and a reference image. This methodology is analogous to the regularization techniques used in AdvART \cite{b124} and DAP \cite{b153}. While we initially considered ensemble methods, their higher computational costs compared to single models led us to adopt a more streamlined approach. Consequently, we chose a lightweight method that preserves the realism and stealthiness of APs while avoiding the computational overhead associated with ensemble strategy.
\begin{definition}[Visual Reality of Adversarial Patches]
	\label{def5}
	Let $\mathcal{P}$ denote the set of image patches and $\mathcal{S}$ represent the set of specified reference images. To ensure visual realism, we aim to minimize the Kullback-Leibler (KL) divergence between $\mathcal{P}$ and $\mathcal{S}$. By employing \textit{Pinsker's inequality}, the lower bound of the KL divergence can be expressed as:
	\begin{equation}
		\mathcal{P} = \arg\min \mathbb{E}_{\hat{x}\sim\mathcal{P}}\left[\mathbb{KL}(\mathcal{P}||\mathcal{S})\right] \geq \frac{1}{2}\mathbb{E}_{\hat{x}\sim\mathcal{P}}\Vert\mathcal{P}-\mathcal{S}\Vert^2.
	\end{equation}
\end{definition}
\begin{theorem}[Upper Bound of Generalization Error]
	\label{th3}
	Consider a lower threshold of generalization error $t$ where $t \rightarrow 0$ and $t > 0$. As defined in Definition \ref{def2}, let $\mathcal{G}(f)$ denote the generalization error of an ensemble of $M$ models evaluated on $N$ data points. The generalization error $\mathcal{G}(f)$ can be bounded with a confidence level of $1-\gamma$, where $\gamma \in (0,1)$, as follows \cite{b116,b154}:
	\begin{equation}
	\mathcal{G}(f)\leq t(M,N,\gamma) = \sqrt{\frac{1}{2N}(\log M + \log \frac{1}{\gamma})}
	\end{equation}
\end{theorem}
Under Theorem \ref{th3}, we augment the dataset with additional transformation operations to enhance data generalization. We employ the Frobenius norm in Hilbert space to design stealthy adversarial attacks, introducing the Compared Specified Image Similarity (CSS) measurement. The CSS metric is defined as follows:
\begin{equation}
	\mathcal{L}_{\text{CSS}} = \mathbb{E}_{p\sim\mathcal{P}}\mathbb{E}_{t\sim\mathcal{T}}\Vert t(p)-t(s)\Vert_{\mathcal{F}}^2,
\end{equation}
where $t\in\mathcal{T}$ represents transformation operations such as rotation, flipping, and cropping. \\
\textbf{Intuition:} There exists an inherent trade-off between the naturalness and the attack efficacy of adversarial patches. Higher attack performance generally results in lower naturalness, manifesting as a reduced resemblance to the source image. Conversely, increased naturalness often corresponds to diminished attack success rates. Further details are provided in Section \ref{stealthiness_evaluation}.
\subsection{Practicality of Adversarial Patches in the Physical Domain}
Previous research has predominantly concentrated on the application of Adversarial Patches (APs) in the digital domain. However, it is essential to recognize that APs can also be effectively utilized in physical settings. To enhance the \textit{practicality} of APs for physical environments, we propose incorporating Total Variation (TV) \cite{b100} and Non-Printable Score (NPS) \cite{b111} as regularization terms. These terms facilitate the creation of more aggressive and robust adversarial patches, specifically tailored to perform effectively in real-world physical environments.

The Total Variation (TV) term is introduced to promote the smoothness of neighboring pixels. The TV regularization is defined as follows:
\begin{equation}
	\mathcal{L_\mathrm{TV}} = \sum_{a,b} \vert p_{a,b}-p_{a,b+1} \vert + \vert p_{a,b}-p_{a+1,b} \vert,
\end{equation}
where $p\in\mathcal{P}$ denotes the adversarial patch, and the indices $a$ and $b$ represent the horizontal and vertical coordinates, respectively.

Additionally, we incorporate the Non-Printable Score (NPS) to enhance the robustness of adversarial patches when transitioning from digital to physical domains, ensuring compatibility with both the iPad color gamut and printer gamut. The NPS is formulated as follows:
\begin{equation}
	\mathcal{L_\mathrm{NPS}}= \sum_{a,b}\vert p_{a,b}-c_{a,b}\vert,
\end{equation}
where $p\in\mathcal{P}$ denotes the adversarial patch and $c\in\mathcal{C}$ represents the color matrix corresponding to the display capabilities of the iPad. The indices $a$ and $b$ refer to the horizontal and vertical coordinates, respectively. To further enhance the physical robustness of the adversarial patches, we employ the Euclidean distance $\vert \cdot \vert$ as a constraint in the generation process.
\subsection{Dual-Perception-Based Framework}
The Dual-Perception-Based Framework (DPBF) integrates the Model-Perception-Based Module (MPBM) and the Human-Perception-Based Module (HPBM) through the incorporation of Total Variation (TV), Non-Printability Score (NPS), and other hyperparameters, addressing the issues delineated in the Motivation section. The MPBM enhances the \textit{transferability} of adversarial patches, while the HPBM ensures these patches are both realistic and \textit{stealthy}. Furthermore, the inclusion of NPS and TV terms significantly augments the \textit{practicality} of adversarial patches (APs), facilitating their seamless application from the digital domain to physical environments.

To expedite the optimization process, additional optimization factors have been introduced into the existing loss functions, culminating in the formulation of an aggregate loss function as presented in Eq. \eqref{e8}:
\begin{equation} 
	\label{e8} 
	\mathcal{L_\mathrm{TOTAL}} = \alpha \mathcal{L_\mathrm{NPS}} + \beta \mathcal{L_\mathrm{TV}} + \lambda \mathcal{L_\mathrm{CSS}} + \mathcal{L_\mathrm{OBJ}}. 
\end{equation}
\begin{algorithm}[htb]
    \caption{Generating Adversarial Patch (MVPatch)}
    \label{a1}
    \begin{algorithmic}[1]      
        \STATE \textbf{Input:} Original Image $\mathcal{I}$, Neural Network $f(x)$
        \STATE \textbf{Output:} Adversarial Patch $\mathcal{P}$
        \STATE Initialize patch mask $\mathcal{M}$
        \STATE Initialize patch $\mathcal{P}$ using $\mathcal{M}$
        \STATE $\mathcal{I^*} = (1 - \mathcal{M}) \odot \mathcal{I} + \mathcal{M} \odot \mathcal{P}$
        \STATE $\mathcal{L}_\mathrm{\text{TOTAL}} \gets \alpha \mathcal{L}_\mathrm{\text{NPS}} + \beta \mathcal{L}_\mathrm{\text{TV}} + \lambda \mathcal{L}_\mathrm{\text{CSS}} + \mathcal{L}_\mathrm{\text{OBJ}}$
        \STATE Initialize learning rate $\mathcal{LR}$ and epochs $E$

        \WHILE{not converged or $E < E_\mathrm{\text{MAX}}$}
            \STATE Update $\mathcal{I^*}$ using backpropagation based on $\mathcal{L}_\mathrm{\text{TOTAL}}$
            \STATE Update patch mask $\mathcal{M}$ based on $\mathcal{I^*}$
            \STATE Update patch $\mathcal{P}$ based on $\mathcal{M}$
            \STATE $E \gets E + 1$
            \IF{parameter optimization threshold not reached}
                \STATE $\mathcal{LR} \gets \gamma \mathcal{  }$
            \ENDIF
        \ENDWHILE

        \STATE \textbf{Return:} Adversarial Patch $P$
    \end{algorithmic}
\end{algorithm}

Since optimizing the ensemble loss function with the Adam optimizer is challenging, we modified our optimization strategy. Instead of reducing the learning rate as training epochs increase, we set thresholds for the parameters and decrease the learning rate when a parameter no longer exhibits optimization during the training process. The updating strategy for the learning rate is governed by Eq. \eqref{e9}:
\begin{equation}
	\label{e9}
	\mathcal{LR_\mathrm{N}} = \gamma \mathcal{LR_\mathrm{O}},
\end{equation} 
where $\mathcal{LR}_\mathrm{N}$ is the new learning rate, while $\mathcal{LR}_\mathrm{O}$ corresponds to the old learning rate. In our experiments, we utilize $\gamma = 0.01$ as the decay coefficient. 

The complete process is summarized in Algorithm \ref{a1}. This algorithm aims to generate an adversarial patch, referred to as MVPatch, for a given original image $\mathcal{I}$ using a neural network $f(x)$ with the output being the adversarial patch $\mathcal{P}$. 
\section{Experiments}\label{s4}
In this section, we provide a brief introduction to the environment and dataset used in the experiment. Additionally, we delineate the threat models considered in the experiment and introduce the evaluation metrics employed. In the digital environment, we generate both meaningful and meaningless adversarial patches and compare them with similar patch attack approaches. In the physical environment, we compare the Attack Success Rate (ASR) of MVPatch with similar patch attack methods and vary the angles and distances between the patch and the camera to investigate their impacts on the ASR.
\subsection{Implementation Details}
In the digital setting, we initialize the adversarial patch using Gaussian noise. We then segment the patch mask onto the person image and overlay the adversarial patch onto the mask. Finally, we input the image with the adversarial patch into the object detection model. The threat models considered in our experiments are YOLOv2, YOLOv3, YOLOv3-tiny, YOLOv4, and YOLOv4-tiny. For evaluating the transferability of the attack, we utilize YOLOv5, Faster-RCNN, and SSD as the transferable attack models. Subsequently, we calculate the loss functions, including $\mathcal{L}_\mathrm{OBJ}$, $\mathcal{L}_\mathrm{NPS}$, $\mathcal{L}_\mathrm{CSS}$, and $\mathcal{L}_\mathrm{TV}$. To optimize the overall loss function $\mathcal{L}_\mathrm{TOTAL} = \alpha \mathcal{L}_\mathrm{NPS} + \beta \mathcal{L}_\mathrm{TV} + \lambda \mathcal{L}_\mathrm{CSS} + \mathcal{L}_\mathrm{OBJ}$ using backpropagation and the Adam optimizer, we set parameters as $\alpha = 0.01$, $\beta = 2.5$, and $\lambda = 2.5$.

In the physical setting, we employ YOLOv2 as the experimental model to evaluate the Attack Success Rate (ASR) of the generated patches. The parameters of YOLOv2, including $\text{Conf}_\text{THRESHOLD}=0.75$ and $\text{NMS}_\text{THRESHOLD}=0.5$, are carefully selected. For collecting the source images, we utilize the iPhone 11 camera, which serves as our image input device.
\subsection{Threat Models}
We apply the YOLOv2, YOLOv3, YOLOv3-tiny, YOLOv4, and YOLOv4-tiny models to generate adversarial patches. These patches are then evaluated on the YOLOv5, Faster-RCNN, and SSD models to assess their transferability in attacks. In the object context of the YOLO series, given the benign image $\mathcal{I}$, the objective of the adversarial attack is to jeopardize the object detector $f_\theta(x,\mathcal{Y})$, where $x$ belongs to $\mathcal{I}$, $\mathcal{Y} = \{\mathrm{a},\mathrm{b},\mathrm{w},\mathrm{h}, \mathrm{obj},\mathrm{cls}\}$ and $\theta$ denotes the parameter of the object detector. We suppress the object confidence score by utilizing the adversarial patches. $\mathrm{a}$ and $\mathrm{b}$ are represented by the value of the horizontal and vertical coordinates of the detection frame. $\mathrm{w}$ and $\mathrm{h}$ are the width and height of the detection frame. $\mathrm{obj}$ and $\mathrm{cls}$ are object confidence scores and classification scores. Our focus is to suppress the person's object confidence score by optimizing the use of adversarial patches.
\subsection{Experimental Environment and Dataset}
For training and evaluating the proposed approach, we utilize the Inria Person dataset\cite{b117}, which is specifically designed for pedestrian detection tasks. The dataset comprises 614 images for training and 288 images for testing. We utilize an RTX 3090 for the computational module and a $12^*$ Xeon Platinum 8260C for task scheduling. Python version 3.6 and PyTorch version 1.6.1 are used.

In the digital setting, we evaluate the performance of the proposed approach using meaningful adversarial patches and meaningless adversarial patches. These patches are added to images of people from the Inria dataset.

In the physical setting, we print the generated MVPatch using an iPad and attach it to the chest of the investigator. By varying the angles ($0^\circ$, $30^\circ$, $60^\circ$, $90^\circ$) and distances (1 meter, 2 meters, 3 meters, 4 meters) between the human and the camera, we investigate the attack performance of MVPatch and other similar attack methods in a physical setting.
\subsection{Evaluation Metrics}
To evaluate the naturalness and attack transferability of MVPatch, we employed the following evaluation metrics: mean Average Precision (mAP), Naturalness Score (NS), and Transferability Score (TS).\\
\textbf{mean Average Precision (mAP):} 
The mAP is a widely used metric calculated by summing the average precision of models for each test sample and dividing it by the total number of test samples. It allows for comparisons of model performance across different datasets. The precision `PREC' is computed as the ratio of true positive samples divided by the sum of true positive and false positive samples. The recall `REC' is calculated as the ratio of true positive samples divided by the sum of true positive and false negative samples. Formally, 
\begin{equation}
	\begin{split}
	\mathrm{PREC}&=\frac{\mathrm{TP}}{\mathrm{TP}+\mathrm{FP}}, \quad \mathrm{REC}=\frac{\mathrm{TP}}{\mathrm{TP}+\mathrm{FN}},   \\
	\mathrm{mAP}&=\frac{1}{n-1}\sum_{i=1}^{n-1}(\mathrm{REC}_{i+1}-\mathrm{REC}_{i})\mathrm{PREC}_{i+1},
	\end{split}
  \end{equation}
where $\text{TP}$, $\text{FP}$, $\text{TN}$, and $\text{FN}$ denote respectively true positive sample, false positive sample, true negative sample, and false negative sample. Through the equation mentioned above, we can obtain the $\text{mAP}$ of the object detection models.\\
\textbf{Naturalness Score (NS):} The NS is used to measure the similarity between the adversarial patch and the specified image, which is defined as
\begin{equation}
	\begin{aligned}
		\mathrm{NS} =  
		& \lambda \left (\frac{\mathrm{CS} - \widehat{\mathrm{CS}}_\mathrm{NG}}{\widehat{\mathrm{CS}}_\mathrm{S}-\widehat{\mathrm{CS}}_\mathrm{NG}} + \frac{\mathrm{CS} - \widehat{\mathrm{CS}}_\mathrm{NR}}{\widehat{\mathrm{CS}}_\mathrm{S}-\widehat{\mathrm{CS}}_\mathrm{NR}}\right )\\
		&+ \left (1-\lambda \right ) \left (\frac{\mathrm{ED} - \widehat{\mathrm{ED}}_\mathrm{NG}}{\widehat{\mathrm{ED}}_\mathrm{NG} - \widehat{\mathrm{ED}}_\mathrm{S}} + \frac{\mathrm{ED} - \widehat{\mathrm{ED}}_\mathrm{NR}}{\widehat{\mathrm{ED}}_\mathrm{NR} - \widehat{\mathrm{ED}}_\mathrm{S}} \right ),
	\end{aligned}
\end{equation}
where $\text{CS}$ denotes the cosine similarity score between the specified image $\text{S}$ and the adversarial patch. $\widehat{\text{CS}}$ signifies the cosine similarity score between the specified image $\text{S}$ and the adversarial patch of grayscale $\text{NG}$ and random noise $\text{NR}$. Furthermore, $\text{ED}$ represents the Euclidean distance. A higher naturalness score implies that the generated adversarial patch exhibits a stronger resemblance to the specified image.\\
\textbf{Transferability Score (TS):} The TS measures the transferable attack performance of the adversarial patch across diverse object detection models, which is defined as
\begin{equation}
	\begin{aligned}
		\mathrm{TS} =  \frac{1}{N}  \sum_{i=0}^{N} \left( \frac{\| D_i - \widehat{D}_{i}^{\mathrm{NG}} \|_2}{\widehat{D}_{i}^{\mathrm{NG}}} + \frac{\| D_i - \widehat{D}_{i}^{\mathrm{NR}} \|_2}{\widehat{D}_{i}^{\mathrm{NR}}}\right),
	\end{aligned}
\end{equation}
where $D_i$ denotes the mean Average Precision (mAP) of object detection models in identifying normal images, whereas $\widehat{D}i$ represents the mAP of object detection models in identifying $\text{NG}$ (gray-scale) and $\text{NR}$ (random noise) images. The symbols $\text{NG}$ and $\text{NR}$ respectively denote grayscale and random noise images. The term $\vert\vert \cdot \vert\vert_2$ refers to the L2 norm, which is employed to constrain the distances between $D_i$ and $\widehat{D}i$. A higher transferability score indicates that the adversarial patch demonstrates superior transferable attack performance. 

The Attack Success Rate (ASR) is defined as follows:
\begin{equation}
	\begin{aligned}
		\mathrm{ASR} =  \frac{1}{\vert \mathcal{I}_\mathrm{test} \vert }\sum_{\mathcal{I}_\mathrm{test}} \mathbbm{1}\{f(\mathcal{I^*,\theta}) \neq f(\mathcal{I},\theta)\},
	\end{aligned}
\end{equation}
where $\mathcal{I}_\mathrm{test}$ denotes a set of benign images derived from test datasets, and $f(\cdot)$ represents the outcome of the models' detection process. The $\text{ASR}$ serves as an indicator of the attack success rate. It equals 1 if $f(\mathcal{I^*,\theta}) \neq f(\mathcal{I},\theta)$, and 0 otherwise.
\section{Results}\label{s5}
\begin{table*}[]
	\centering
	\renewcommand{\arraystretch}{1.5}
	\caption{Attack Performance and Comprehensive Evaluation with Diverse Objective Loss Functions}
	\label{tab1}
	\resizebox{\textwidth}{!}{
		\begin{tabular}{cccccccccccc}
			
			\toprule
			& Faster-RCNN & SSD & YOLOv2 & YOLOv3 & YOLOv3-tiny & YOLOv4 & YOLOv4-tiny & YOLOv5 & Mean mAP & CS & ED\\\midrule
			\textbf{$\mathcal{L}_{\text{OBJ}_1}$} & \textbf{42.07\%} & \textbf{54.57\%} & \textbf{30.26\%} & \textbf{31.39\%}&\textbf{13.49\%}&\textbf{27.31\%}&\textbf{22.74\%}&\textbf{35.72\%} & \textbf{32.20\%} & \textbf{97.91\%} & \textbf{0.0428}\\
			$\mathcal{L}_{\text{OBJ}_2}$&	49.40\%&	61.08\%&	50.34\%&	37.47\%&	41.44\%&	52.64\%&	56.46\%&56.46\%  & 51.80\% & 97.91\% & 0.0419\\\bottomrule	
			
	\end{tabular}}
\end{table*}
In this section, we present empirical evidence that underscores the efficacy of the proposed MVPatch through extensive evaluations conducted in both digital and physical domains. Initially, we illustrate the experimental results of various ensemble model strategies, demonstrating that our proposed ensemble model strategy exhibits superior performance compared to alternative approaches. Additionally, we evaluate naturalness by varying the natural factor $\lambda$. Within the digital domain, we categorize the experimental subjects into meaningful and meaningless adversarial patches. In the physical domain, we assess the performance of the patch attack from different angles and distances by employing a camera to capture images as input sources. The experimental outcomes reveal that our method surpasses other comparable adversarial patch attack methods in both digital and physical domains. Lastly, we present the visual outcomes of the experiments conducted in both the digital and physical domains.
\subsection{Evaluation of Transferability with Diverse Ensemble Model Strategies}
Based on theoretical analysis discussed in Section \ref{MPBM} concerning the enhancement of transferability, the ensemble modeling technique can achieve superior adversarial transferability performance across multiple models when compared to a single model. There are, however, two primary approaches for leveraging the ensemble method: maximizing the outputs of neural networks and averaging the outputs of neural networks. We first investigate the utilization of two distinct loss functions for attacking object detection models, as outlined below:
\begin{table}[htbp]
	\centering
	\renewcommand{\arraystretch}{1.5}
	\caption{Experimental Consequence of Naturalness of MVPatch}
	\label{t3}
	\begin{tabular}{ccccc}
		\hline
		\multirow{2}{*}{Adversarial Patches} & \multicolumn{4}{c}{Evaluation Metrics} \\
		& mAP       & CS       & ED    & NS      \\ \hline
		SRC                                & 65.70\%   & 100\%    & 0.00  & 100.00  \\
		NoiseGrey                          & 74.05\%   & 97.45\%  & 6.01  & 0.00    \\
		NoiseRandom                        & 75.09\%   & 92.82\%  & 8.27  & 0.00    \\
		$\lambda=1$                        & 32.30\%   & 97.61\%  & 4.63  & 35.59   \\
		$\lambda=1.25$                     & 32.19\%   & 97.87\%  & 4.24  & 42.11   \\
		$\lambda=1.5$                      & 34.39\%   & 98.54\%  & 3.73  & 56.77   \\
		$\lambda=1.75$                     & 33.13\%   & 98.78\%  & 3.44  & 62.48   \\
		$\lambda=2$                        & 33.98\%   & 99.01\%  & 3.12  & 68.14   \\
		$\lambda=2.25$                     & 37.49\%   & 99.20\%  & 2.82  & 72.97   \\
		$\bm{\lambda=2.5}$                      & \textbf{36.63\%}   & \textbf{99.30\%}  & \textbf{2.53}  & \textbf{76.08}  \\
		$\lambda=2.75$                     & 37.46\%   & 99.38\%  & 2.21  & 78.94   \\
		$\lambda=3$                        & 39.69\%   & 99.44\%  & 1.99  & 81.01   \\ \hline
	\end{tabular}
\end{table}
\begin{equation}
	\mathcal{L}_\mathrm{OBJ_1} = \frac{1}{M}\sum_{i=1}^{M}f_i(\hat{x}),
\end{equation}
\begin{equation}
	\mathcal{L}_\mathrm{OBJ_2} = \max \{f_i(\hat{x})\}^{M}_{i=1},
\end{equation}
where $f(\cdot)$ denotes the object confidence score of the $i^{th}$ object detection model in the ensemble of object detection models (YOLOv2, YOLOv3, YOLOv3-tiny, YOLOv4, YOLOv4-tiny). $\hat{x}\in\mathcal{I^*}$ is the image with adversarial patch.

We compare the attack performance of diverse object confidence score loss functions on various object detection models. It is observed that the adversarial patch generated by $\mathcal{L}_\mathrm{{OBJ}_1}$ outperforms the one generated by $\mathcal{L}_\mathrm{{OBJ}_2}$ in terms of attacking the object detection model. The adversarial patches generated through diverse loss functions are then applied to eight object detection models (Faster-RCNN, SSD, YOLOv2, YOLOv3, YOLOv3-tiny, YOLOv4, YOLOv4-tiny, and YOLOv5 models) to assess their attack performance and similarity to specified images. The adversarial patch generated by $\mathcal{L}_\mathrm{{OBJ}_1}$ not only exhibits excellent attack capability (the average mAP of $\mathcal{L}_\mathrm{{OBJ}_1}$ is 32.20\% and the average mAP of $\mathcal{L}_\mathrm{{OBJ}_2}$ is 51.80\%) but also maintains a high level of naturalness. Consequently, we select $\mathcal{L}_\mathrm{{OBJ}_1}$ as our objective loss function.
\begin{figure}[htbp]
	\centering
	\includegraphics[width=1\linewidth]{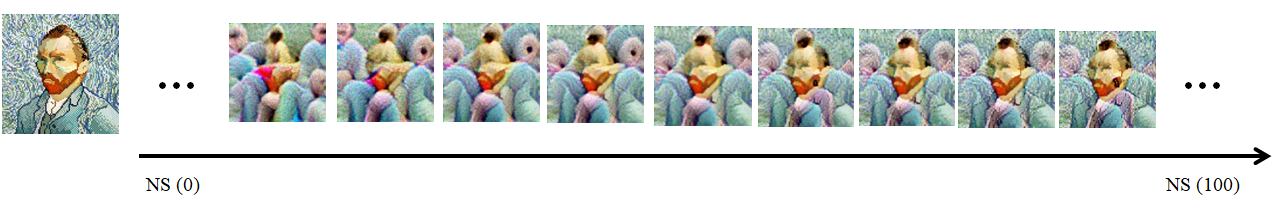}
	\caption{The naturalness score (NS) of adversarial patches with diverse $\lambda$ parameters. As the NS increases, the level of similarity between the source image and the generated image rises.}
	\label{f4}
\end{figure}

The results of the experiment are illustrated in Table \ref{tab1}, where the Faster-RCNN, SSD, YOLOv2, YOLOv3, YOLOv3-tiny, YOLOv4, YOLOv4-tiny, and YOLOv5 models are included as threat models. The mean mAP represents the average value of the mean Average Precision (mAP) obtained from diverse models. CS corresponds to the Cosine Similarity and ED denotes the Euclidean Distance.
\subsection{Evaluation of Naturalness with Diverse Natural Factors}
\label{stealthiness_evaluation}
We assessed adversarial patches generated using YOLO models, adjusting the natural factor $\lambda$ within [1,3], and found that increasing $\lambda$ enhances naturalness, measured by the naturalness score (NS). The attack performance and similarity to the source image (SRC) were evaluated using mean Average Precision (mAP), Cosine Similarity (CS), and Euclidean Distance (ED). Additionally, $\text{NoiseGrey}$ and $\text{NoiseRandom}$ are matrices with values of 0.5 and random values, respectively.

As shown in Table \ref{t3}, varying $\lambda$ from 1 to 3 increases the naturalness score (NS) from 35.59 to 81.01. Simultaneously, the mean Average Precision (mAP) rises from 32.30\% to 39.69\%. However, a higher mAP generally indicates greater accuracy of the object detection models, leading to a lower attack performance of the adversarial patches. This highlights a tradeoff between naturalness and attack performance, making it challenging to achieve both high naturalness and high attack performance in a single adversarial patch.

Through analysis, we identify a threshold for naturalness that optimizes both aspects. Feedback from 20 experimental participants indicates that a naturalness score above 75 (achieved with $\lambda=2.5$) is generally perceived as more natural. Beyond this point, further increases in $\lambda$ result in diminishing gains in naturalness as perceived by the human eye. Therefore, we define $\lambda=2.5$ as the optimal parameter for the loss function in subsequent experiments. Figure \ref{f4} illustrates images of adversarial patches with different naturalness scores.
\begin{table*}[]
	\renewcommand{\arraystretch}{1.5}
	\centering
	\caption{Experimental Consequence of Transferability of Meaningful Adversarial Patches (mAP \%)}
	\label{t4}
	\scalebox{0.83}{
	\begin{tabular}{cccccccccc}
		\hline
		\multirow{2}{*}{$^{(Meaningful Adversarial Patches)}$Threat Models}         & \multicolumn{8}{c}{Transfer Attack Models}                                                                                            & \multirow{2}{*}{Black Box(TS $\uparrow$)} \\
		& Faster RCNN    & SSD            & YOLOv2         & YOLOv3         & YOLOv3-tiny    & YOLOv4         & YOLOv4-tiny    & YOLOv5         &                                        \\ \hline
		$^{\mathrm{(Fig.\ref{fa})}}$ YOLOv2 {[}Ours{]}                   & 55.75          & 63.30          & 32.28          & 60.37          & 29.93          & 59.98          & 51.79          & 62.80          & 20.22                                  \\
		$^{\mathrm{(Fig.\ref{fi})}}$ YOLOv2\cite{b96}                             & 56.60          & 56.66          & 37.96          & 56.85          & 58.04          & 67.74          & 67.43          & 66.85          & 15.21                                  \\
		$^{\mathrm{(Fig.\ref{fb})}}$ YOLOv3 {[}Ours{]}                   & 50.81          & 60.53          & 50.55          & 55.44          & 41.54          & 52.80          & 56.79          & 61.49          & 27.93                                  \\
		$^{\mathrm{(Fig.\ref{fj})}}$ YOLOv3\cite{b96}                             & 55.35          & 51.15          & 49.44          & 55.39          & 52.10          & 66.97          & 67.09          & 58.69          & 23.53                                  \\
		$^{\mathrm{(Fig.\ref{fc})}}$ YOLOv3-tiny {[}Ours{]}              & 59.77          & 67.39          & 59.92          & 69.63          & 5.31           & 71.32          & 62.45          & 74.66          & 21.14                                  \\
		$^{\mathrm{(Fig.\ref{fk})}}$ YOLOv3-tiny\cite{b96}                        & 52.81          & 51.55          & 48.75          & 63.26          & 38.99          & 62.59          & 65.93          & 64.11          & 24.9                                   \\
		$^{\mathrm{(Fig.\ref{fd})}}$ YOLOv4 {[}Ours{]}                   & 47.32          & 56.91          & 44.96          & 52.04          & 15.35          & 27.89          & 46.37          & 45.11          & 43.69                                  \\
		$^{\mathrm{(Fig.\ref{fl})}}$ YOLOv4\cite{b96}                             & 57.70          & 60.91          & 58.73          & 66.08          & 69.39          & 72.64          & 71.13          & 75.76          & 10.76                                  \\
		$^{\mathrm{(Fig.\ref{fe})}}$ YOLOv4-tiny {[}Ours{]}             & 60.85          & 66.16          & 49.25          & 68.07          & 47.26          & 72.87          & 29.53          & 71.32          & 22                                     \\
		$^{\mathrm{(Fig.\ref{fm})}}$ YOLOv4-tiny\cite{b96}                        & 54.64          & 41.15          & 39.90          & 50.28          & 31.39          & 61.88          & 57.45          & 54.61          & 34.41                                  \\
		\textbf{$\bm{^{\mathrm{(-)}}}$Ensemble Model\cite{b96}}                        & \textbf{61.28}          & \textbf{52.28}          & \textbf{49.42}          & \textbf{35.46}          & \textbf{25.29}          & \textbf{51.71}          & \textbf{18.51}          & \textbf{64.00}          & \textbf{40}                                     \\
		\textbf{$\bm{^{\mathrm{(Fig.\ref{fn})}}}$ Ensemble Model {[}Ours{]}} & \textbf{45.96} & \textbf{55.52} & \textbf{32.09} & \textbf{25.66} & \textbf{21.69} & \textbf{36.43} & \textbf{28.74} & \textbf{35.36} & \textbf{52.82}                         \\ 
		\textbf{$\bm{^{\mathrm{(Fig.\ref{ff})}}}$ Ensemble Model {[}Ours{]}} & \textbf{42.07} & \textbf{54.57} & \textbf{30.26} & \textbf{31.39} & \textbf{13.49} & \textbf{27.31} & \textbf{22.74} & \textbf{35.72} & \textbf{56.83}                         \\\hline
		$^{\mathrm{(Fig.\ref{fg},\ref{fo})}}$ Srouce Image                       & 60.66          & 66.76          & 58.60          & 68.59          & 57.04          & 68.78          & 70.08          & 75.07          & 0                                      \\
		$^{\mathrm{(Fig.\ref{fh})}}$ Grey-scale Noise                    & 61.75          & 72.05          & 67.75          & 76.22          & 80.69          & 75.22          & 76.89          & 81.86          & 0                                      \\
		$^{\mathrm{(Fig.\ref{fp})}}$ Random Noise                        & 63.70          & 73.19          & 69.67          & 75.37          & 82.36          & 75.79          & 78.95          & 81.69          & 0                                      \\ \hline
	\end{tabular}}
\end{table*}
\begin{figure*}[!htb]
	\centering
	\subfigure[]{\includegraphics[width=0.1\hsize, height=0.1\hsize]{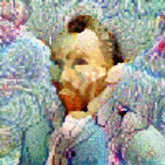}\label{fa}}\hspace{1cm}
	\subfigure[]{\includegraphics[width=0.1\hsize, height=0.1\hsize]{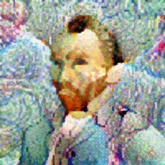}\label{fb}}
	\subfigure[]{\includegraphics[width=0.1\hsize, height=0.1\hsize]{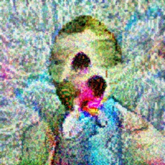}\label{fc}}
	\subfigure[]{\includegraphics[width=0.1\hsize, height=0.1\hsize]{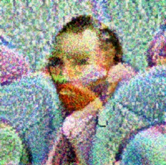}\label{fd}}
	\subfigure[]{\includegraphics[width=0.1\hsize, height=0.1\hsize]{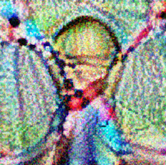}\label{fe}}
	\subfigure[]{\includegraphics[width=0.1\hsize, height=0.1\hsize]{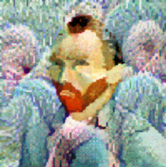}\label{ff}}
	\subfigure[]{\includegraphics[width=0.1\hsize, height=0.1\hsize]{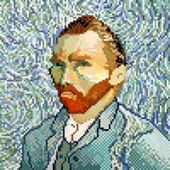}\label{fg}}
	\subfigure[]{\includegraphics[width=0.1\hsize, height=0.1\hsize]{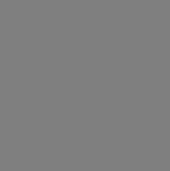}\label{fh}}
	\vspace{1cm}
	\subfigure[]{\includegraphics[width=0.1\hsize, height=0.1\hsize]{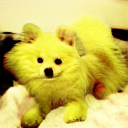}\label{fi}}\hspace{1cm}
	\subfigure[]{\includegraphics[width=0.1\hsize, height=0.1\hsize]{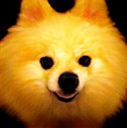}\label{fj}}
	\subfigure[]{\includegraphics[width=0.1\hsize, height=0.1\hsize]{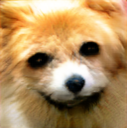}\label{fk}}
	\subfigure[]{\includegraphics[width=0.1\hsize, height=0.1\hsize]{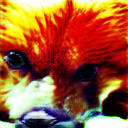}\label{fl}}
	\subfigure[]{\includegraphics[width=0.1\hsize, height=0.1\hsize]{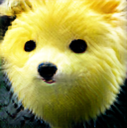}\label{fm}}
	\subfigure[]{\includegraphics[width=0.1\hsize, height=0.1\hsize]{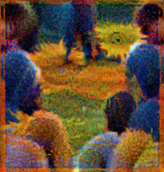}\label{fn}}
	\subfigure[]{\includegraphics[width=0.1\hsize, height=0.1\hsize]{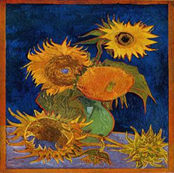}\label{fo}}
	\subfigure[]{\includegraphics[width=0.1\hsize, height=0.1\hsize]{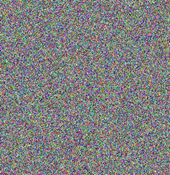}\label{fp}}
	\caption{Adversarial patches are generated by various object detectors, including YOLOv2, YOLOv3, YOLOv3-tiny, YOLOv4, and YOLOv4-tiny and compared attack performance on the individual model with Natural Patch\cite{b96}. Additionally, we compare the performance of our ensemble attack with that of the Natural Patch. The attack performance of diverse adversarial patches are illustrated as Table \ref{t4}.}
	\label{f5}
\end{figure*}
\begin{table*}[h]
	\centering
	\renewcommand{\arraystretch}{1.5}
	\caption{Experimental Consequence of Transferability of Meaningless Adversarial Patches (mAP \%)}
	\label{t5}
	\resizebox{\textwidth}{!}{
	\begin{tabular}{cccccccccc}
		\hline
		\multirow{2}{*}{Meaningless Adversarial Patches} & \multicolumn{8}{c}{Transfer Models}                                                                                                  & \multirow{2}{*}{Black Box(TS$\uparrow$)} \\
		& Faster-RCNN    & SSD            & YOLOv2         & YOLOv3         & YOLOv3-tiny   & YOLOv4         & YOLOv4-tiny    & YOLOv5         &                                   \\ \hline
		AdvPatch\cite{b87}            & 43.11          & 48.57          & 4.69           & 47.59          & 34.24         & 58.74          & 33.82          & 49.65          & 46.29                             \\
		AdvTexture\cite{b100}         & 48.18          & 36.12          & 5.31           & 46.83          & 18.74         & 63.03          & 34.13          & 60.47          & 47.56                             \\
		AdvCloak\cite{b99}            & 55.19          & 60.82          & 33.74          & 54.77          & 53.42         & 67.57          & 56.12          & 68.05          & 24.62                             \\ 
		\textbf{MVPatch{[}Ours{]}}                     & \textbf{34.95} & \textbf{48.95} & \textbf{23.93} & \textbf{13.09} & \textbf{6.04} & \textbf{26.72} & \textbf{14.77} & \textbf{34.35} & \textbf{64.47}                    \\\hline
	\end{tabular}}
\end{table*}
\begin{table*}[]
	\centering
	\renewcommand{\arraystretch}{1.5}
	\caption{Experimental Consequence of MVPatch and other Patches in the Physical Domain}
	\label{t6}
	\begin{tabular}{ccccccc}
	\hline
	\multirow{2}{*}{Adversarial Patches} & \multicolumn{6}{c}{Evaluation Metrics}                                  \\
										 & ASR   & \multicolumn{1}{l}{Number of Images}& Stealthiness Attack & Physical Attack & \multicolumn{1}{l}{Transferable Attack} & \multicolumn{1}{l}{Lightweight Model} \\ \hline
	$^{\mathrm{(Fig.\ref{f6a})}}$AdvPatch\cite{b87}                              & 13.57\% & 1319&  $\times$                         & \checkmark         &$\times$&\checkmark                   \\
	$^{\mathrm{(Fig.\ref{f6b})}}$AdvTexture\cite{b100}                           & 12.72\% & 1211 & $\times$                               & \checkmark         &$\times$&    $\times$                  \\
	$^{\mathrm{(Fig.\ref{f6c})}}$Benign Image                        & 5.03\% & 1194& \checkmark                              & $\times$&$\times$&$\times$                                \\
	$^{\mathrm{(Fig.\ref{f6d})}}$Natural Patch\cite{b96}                    & 8.32\% & 1646 & \checkmark                              &\checkmark     &\checkmark &  $\times$                        \\
	$^{\mathrm{(Fig.\ref{f6e})}}$Natural Patch\cite{b96}                 & 19.43\%  & 1616 & \checkmark                              &  \checkmark       &\checkmark&    $\times$                   \\
	\textbf{$^{\mathrm{(Fig.\ref{f6f})}}$MVPatch {[}Ours{]}}                                & \textbf{22.60\%}  & \textbf{1438}& \textbf{\checkmark }                             &\textbf{\checkmark }             &\textbf{\checkmark}& \textbf{\checkmark }                 \\
	\textbf{$^{\mathrm{(Fig.\ref{f6g})}}$MVPatch {[}Ours{]}}                             & \textbf{26.33\%} & \textbf{1257} & \textbf{\checkmark}                              &\textbf{\checkmark}         &\textbf{\checkmark}&     \textbf{\checkmark  }                \\ \hline
	\end{tabular}
	\end{table*}
	\begin{figure*}[htb]
		\centering
		\subfigure[]{\includegraphics[width=0.1\hsize, height=0.1\hsize]{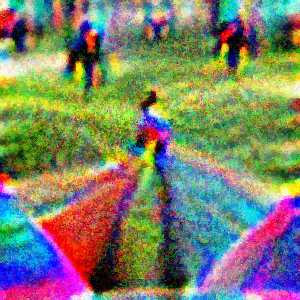}\label{f6a}}
		\subfigure[]{\includegraphics[width=0.1\hsize, height=0.1\hsize]{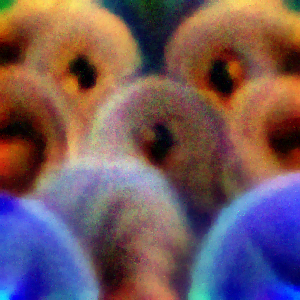}\label{f6b}}
		\subfigure[]{\includegraphics[width=0.1\hsize, height=0.1\hsize]{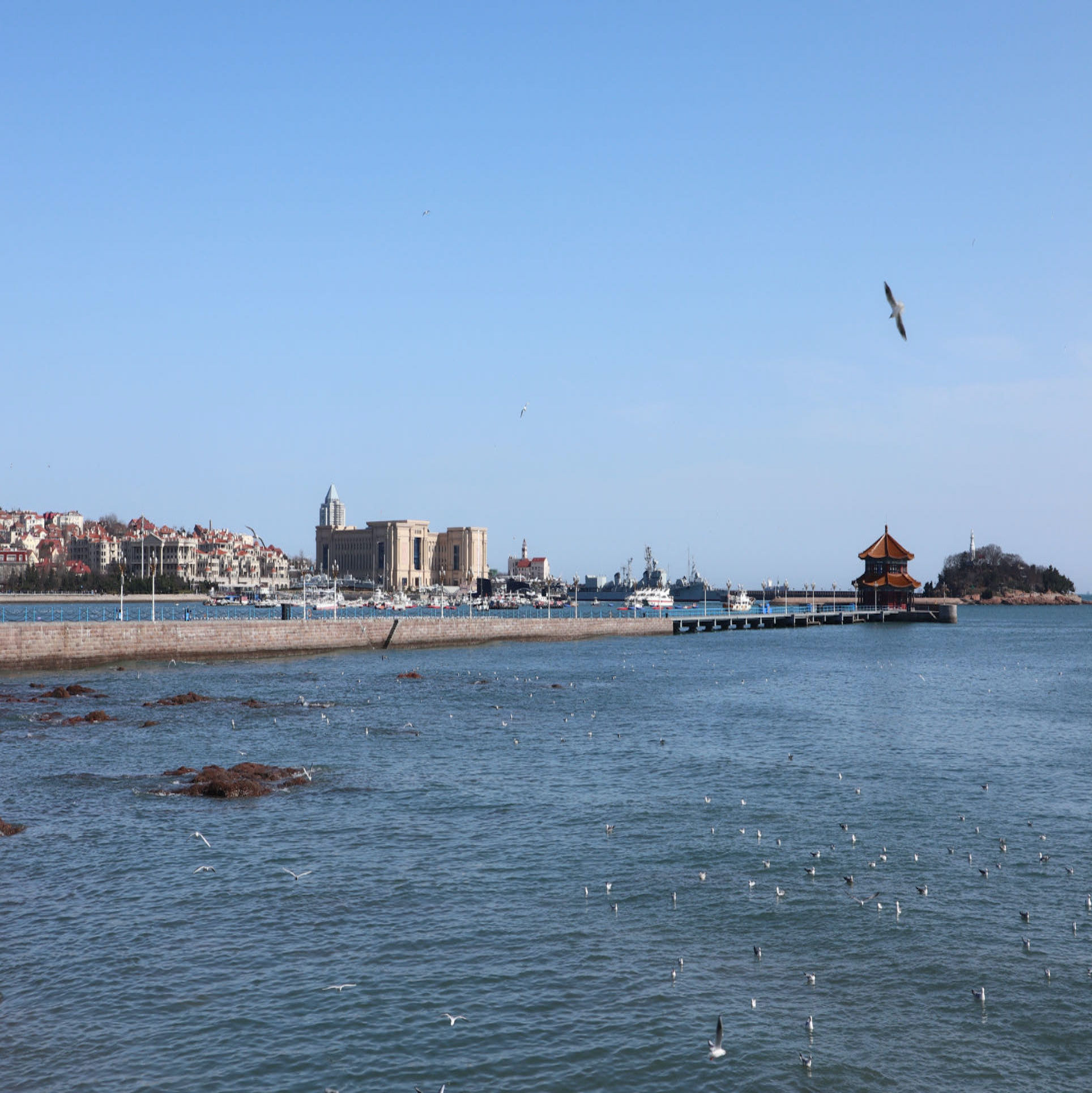}\label{f6c}}
		\subfigure[]{\includegraphics[width=0.1\hsize, height=0.1\hsize]{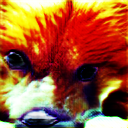}\label{f6d}}
		\subfigure[]{\includegraphics[width=0.1\hsize, height=0.1\hsize]{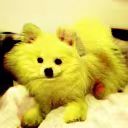}\label{f6e}}
		\subfigure[]{\includegraphics[width=0.1\hsize, height=0.1\hsize]{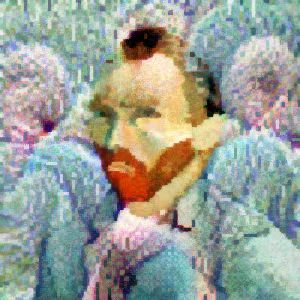}\label{f6f}}
		\subfigure[]{\includegraphics[width=0.1\hsize, height=0.1\hsize]{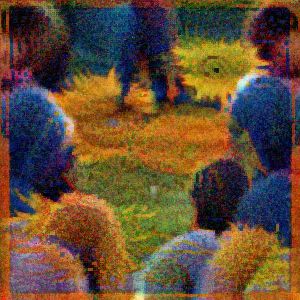}\label{f6g}}
		\caption{Adversarial patches, such as AdvPatch\cite{b87}, AdvTexture\cite{b100}, Natural Patch\cite{b96}, and our MVPatch, are employed in the physical world using diverse attack methods. Table\ref{t6} displays the attack success rate of different adversarial patches, as well as other important evaluation metrics.}
		\label{f6}
	\end{figure*}
	\begin{figure*}[htb]
		\centering
		\subfigure[Impact of diverse angles to ASR of adversarial patches]{\includegraphics[width=0.48\linewidth]{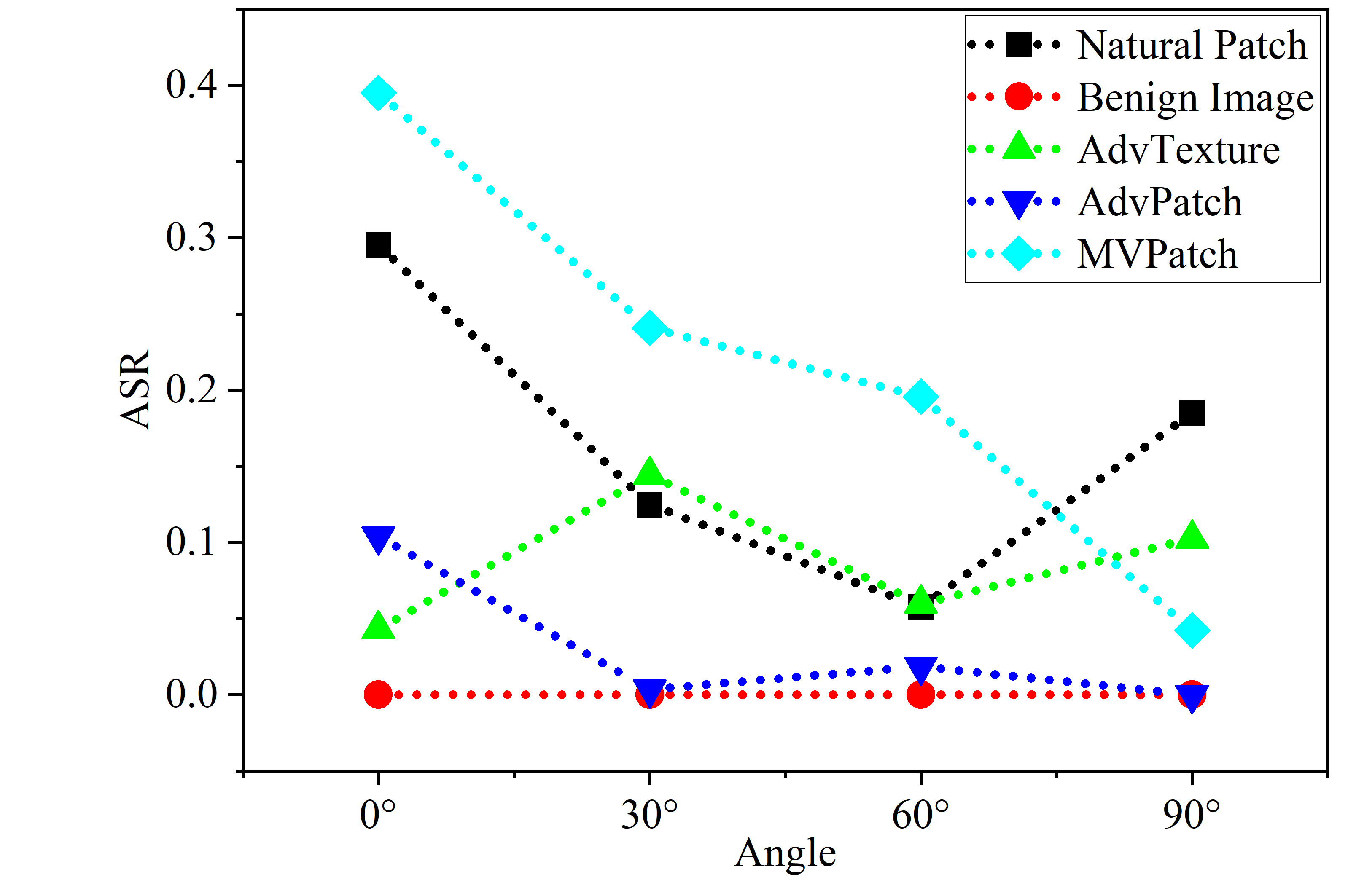}\label{f7a}}
		\subfigure[Impact of diverse distances to ASR of adversarial patches]{\includegraphics[width=0.48\linewidth]{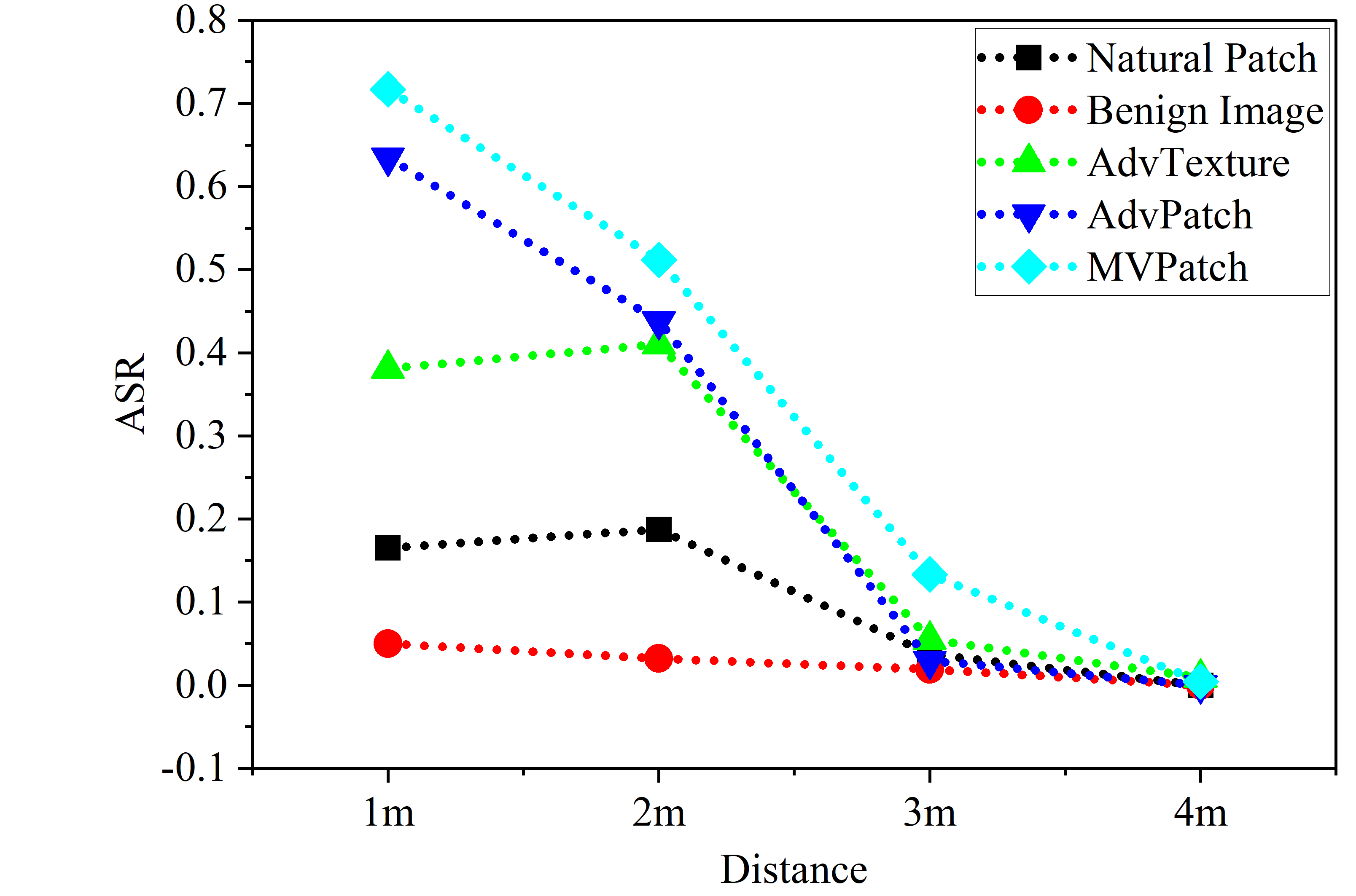}\label{f7b}}
		\caption{Experimental consequence illustration of the adversarial patches attacking in the physical domain. Fig.\ref{f7a} illustrates the impact of varying camera angles on the ASR of different attack methods, such as Natural Patch, AdvPatch, AdvTexture, Benign Image, and MVPatch. Fig.\ref{f7b} illustrates the impact of varying distances between the camera and patches on the ASR of different attack methods, such as Natural Patch, AdvPatch, AdvTexture, Benign Image, and MVPatch.}
		\label{f7}
	\end{figure*}
\subsection{Evaluation in the Digital Domain}
To systematically investigate the attack performance of adversarial patches, we categorize the experiments into two parts: meaningful and meaningless adversarial patches. Most research has primarily focused on meaningless patches, while the exploration of both meaningful and meaningless patches in digital and physical domains remains limited.
\subsubsection{Meaningful Adversarial Patches}
We utilize pre-trained models, including YOLOv2, YOLOv3, YOLOv3-tiny, YOLOv4, and YOLOv4-tiny, to generate diverse adversarial patches with meaningful content using a natural factor of $\lambda=2.5$, as shown in Figs. \ref{fa} to \ref{fe}. By combining these models, we create ensemble attack adversarial patches (Figs. \ref{ff} and \ref{fn}). To compare the effects of person images and other images, we use a sunflower image as a control for the person image. All generated patches are tested against transfer attack models, including Faster R-CNN, SSD, YOLOv2, YOLOv3, YOLOv3-tiny, YOLOv4, YOLOv4-tiny, and YOLOv5, to calculate transferability scores (TS) and compare them with Natural Patch \cite{b96}. The results are presented in Table \ref{t4}. The "Threat Models" column indicates adversarial patches generated by various models, while "Grey-scale Noise" and "Random Noise" refer to images padded with a value of 0.5 and random values, respectively. Higher TS indicates better transferability of the adversarial attacks.

\textbf{Analysis and Discussion:}
We compare the transferability score (TS) of MVPatch (Ours) to that of Natural Patch \cite{b96} on individual object detection models. MVPatch performs better than Natural Patch on most models in the experiment. For the YOLOv4 model, the TS of MVPatch is approximately 33\% higher than that of Natural Patch. Similarly, MVPatch demonstrates excellent attack performance on other object detection models, including YOLOv2, YOLOv3, YOLOv3-tiny, and YOLOv4-tiny. Additionally, MVPatch surpasses Natural Patch in terms of attack performance on the ensemble model, with the TS of MVPatch being approximately 16\% higher than that of Natural Patch. The results of the meaningful adversarial patch experiment clearly indicate that the proposed algorithm, MVPatch, achieves superior performance compared to the similar algorithm, Natural Patch, in terms of attack transferability.
\subsubsection{Meaningless Adversarial Patches}
We set the natural factor to $\lambda=0$ and employ joint pre-trained models, including YOLOv2, YOLOv3, YOLOv3-tiny, YOLOv4, and YOLOv4-tiny, to generate the meaningless adversarial patch. We compare this with AdvPatch \cite{b87}, AdvTexture \cite{b100}, and AdvCloak \cite{b99}. All generated patches are tested on Faster R-CNN, SSD, YOLOv2, YOLOv3, YOLOv3-tiny, YOLOv4, YOLOv4-tiny, and YOLOv5 to calculate transferability scores (TS), which are then compared with the benchmarks. The results are summarized in Table \ref{t5}. The "Meaningless Adversarial Patches" column lists patches generated by various attack algorithms, while the Black Box (TS) indicates transferability, with higher scores reflecting superior performance.

\textbf{Analysis and Discussion:} As illustrated in Table \ref{t5}, the transferability score (TS) of MVPatch is approximately 40\% higher than that of AdvCloak. Compared to AdvPatch and AdvTexture, the TS of MVPatch is approximately 20\% higher. Moreover, when the transfer models, including YOLOv2, YOLOv3, YOLOv3-tiny, YOLOv4, YOLOv4-tiny, YOLOv5, Faster R-CNN, and SSD, are attacked by MVPatch, their mean Average Precision (mAP) reduces by 10 to 20\%, compared to the other three algorithms. To sum up, in terms of meaningless adversarial patches, the proposed algorithm, MVPatch, exhibits higher attack transferability compared to similar algorithms.
\subsection{Evaluation in the Physical Domain}
To systematically investigate the attack performance of adversarial patches in the physical domain, we divide the experimental conditions into two parts: diverse angles and diverse distances. Prior research mainly focused on the attack performance of adversarial patches in the physical domain and few researchers investigated how impact the attack performance of adversarial patches using diverse angles and distances. 
\begin{figure*}[!htb]
	\centering
	\subfigure[Benign Images]{
	\begin{minipage}[t]{0.13\linewidth}
	\includegraphics[width=1\linewidth]{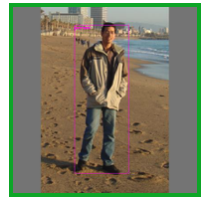}\vspace{1pt}
	\includegraphics[width=1\linewidth]{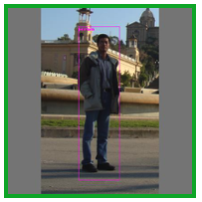}\vspace{1pt}
	\includegraphics[width=1\linewidth]{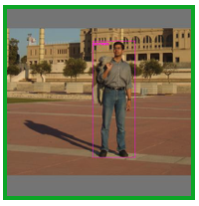}\vspace{1pt}
	\includegraphics[width=1\linewidth]{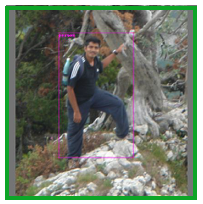}\vspace{1pt}
	\includegraphics[width=1\linewidth]{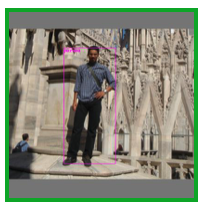}\vspace{1pt}
	\includegraphics[width=1\linewidth]{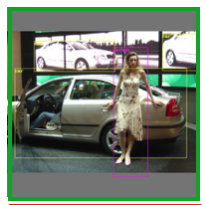}\vspace{1pt}
	\end{minipage}}
	\subfigure[Malicious Images]{
	\begin{minipage}[t]{0.13\linewidth}
	\includegraphics[width=1\linewidth]{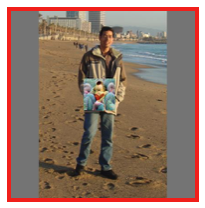}\vspace{1pt}
	\includegraphics[width=1\linewidth]{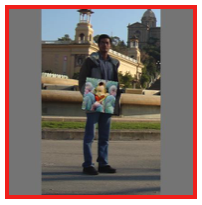}\vspace{1pt}
	\includegraphics[width=1\linewidth]{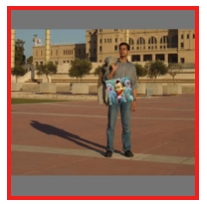}\vspace{2pt}
	\includegraphics[width=1\linewidth]{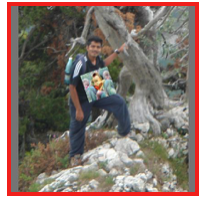}\vspace{2pt}
	\includegraphics[width=1\linewidth]{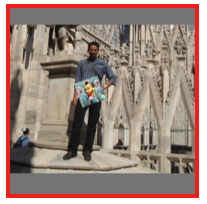}\vspace{2pt}
	\includegraphics[width=1\linewidth]{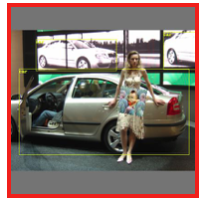}\vspace{2pt}
	\end{minipage}}
	\subfigure[Benign Images]{
	\begin{minipage}[t]{0.13\linewidth}
	\includegraphics[width=1\linewidth]{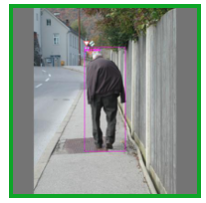}\vspace{1pt}
	\includegraphics[width=1\linewidth]{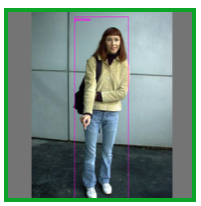}\vspace{1pt}
	\includegraphics[width=1\linewidth]{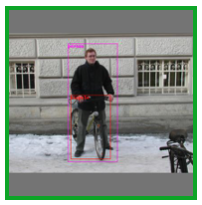}\vspace{1pt}
	\includegraphics[width=1\linewidth]{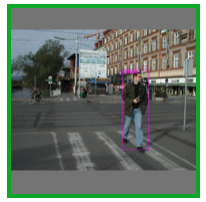}\vspace{1pt}
	\includegraphics[width=1\linewidth]{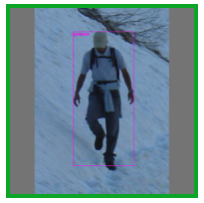}\vspace{1pt}
	\includegraphics[width=1\linewidth]{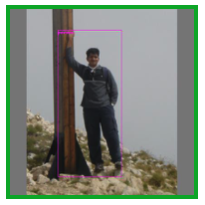}\vspace{1pt}
	\end{minipage}}
	\subfigure[Malicious Images]{
	\begin{minipage}[t]{0.13\linewidth}
	\includegraphics[width=1\linewidth]{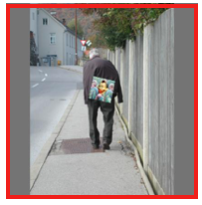}\vspace{1pt}
	\includegraphics[width=1\linewidth]{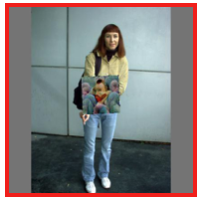}\vspace{2pt}
	\includegraphics[width=1\linewidth]{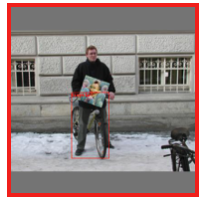}\vspace{2pt}
	\includegraphics[width=1\linewidth]{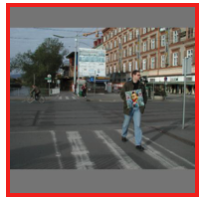}\vspace{2pt}
	\includegraphics[width=1\linewidth]{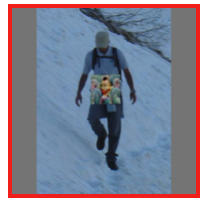}\vspace{2pt}
	\includegraphics[width=1\linewidth]{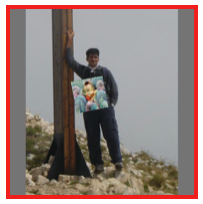}\vspace{2pt}
	\end{minipage}}
	\subfigure[Benign Images]{
	\begin{minipage}[t]{0.13\linewidth}
	\includegraphics[width=1\linewidth]{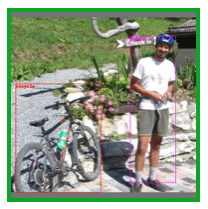}\vspace{1pt}
	\includegraphics[width=1\linewidth]{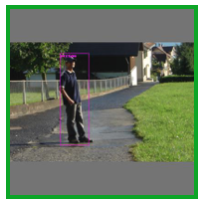}\vspace{1pt}
	\includegraphics[width=1\linewidth]{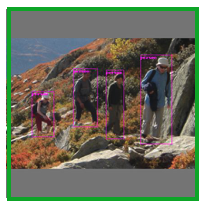}\vspace{1pt}
	\includegraphics[width=1\linewidth]{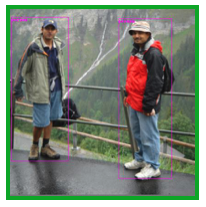}\vspace{1pt}
	\includegraphics[width=1\linewidth]{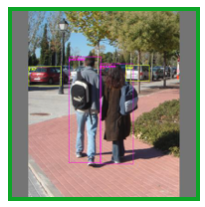}\vspace{1pt}
	\includegraphics[width=1\linewidth]{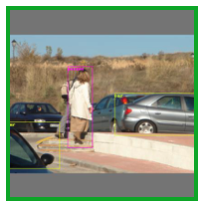}\vspace{1pt}
	\end{minipage}}
	\subfigure[Malicious Image]{
	\begin{minipage}[t]{0.13\linewidth}
	\includegraphics[width=1\linewidth]{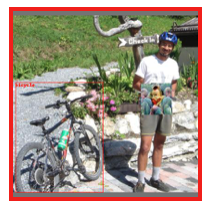}\vspace{1pt}
	\includegraphics[width=1\linewidth]{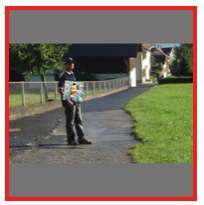}\vspace{1pt}
	\includegraphics[width=1\linewidth]{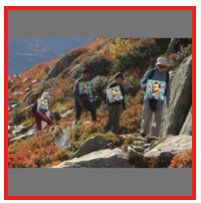}\vspace{1pt}
	\includegraphics[width=1\linewidth]{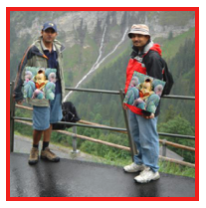}\vspace{1pt}
	\includegraphics[width=1\linewidth]{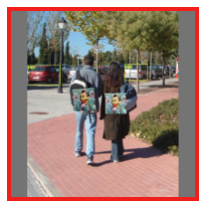}\vspace{1pt}
	\includegraphics[width=1\linewidth]{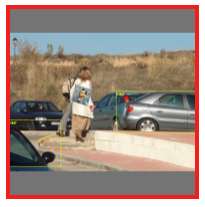}\vspace{1pt}
	\end{minipage}}
	\caption{Illustration of experimental consequence in the digital domain. The images with green rectangles represent the object detectors can recognize the class of a person successfully while the images with red rectangles denote the object detectors can not recognize the class of a person.}
	\label{f9}
	\end{figure*}
\begin{figure*}[!htb]
	\centering
	\subfigure[MVPatch]{
	\begin{minipage}[t]{0.11\linewidth}
	\includegraphics[width=1\linewidth]{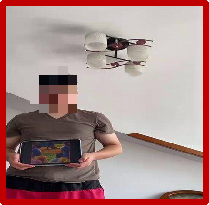}\vspace{1pt}
	\includegraphics[width=1\linewidth]{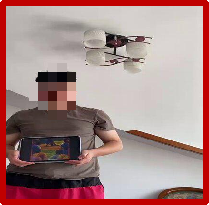}\vspace{1pt}
	\includegraphics[width=1\linewidth]{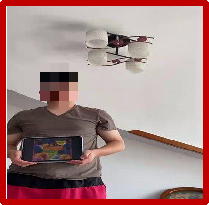}\vspace{1pt}
	\includegraphics[width=1\linewidth]{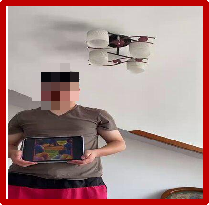}\vspace{1pt}
	\includegraphics[width=1\linewidth]{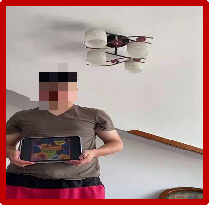}\vspace{1pt}
	\includegraphics[width=1\linewidth]{phyexp-1.jpg}\vspace{1pt}
	\end{minipage}}
	\subfigure[MVPatch]{
	\begin{minipage}[t]{0.11\linewidth}
	\includegraphics[width=1\linewidth]{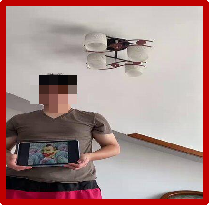}\vspace{1pt}
	\includegraphics[width=1\linewidth]{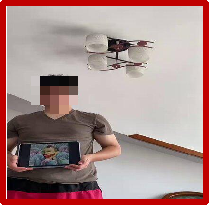}\vspace{1pt}
	\includegraphics[width=1\linewidth]{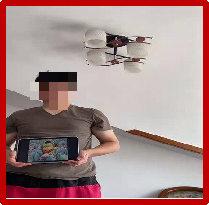}\vspace{1pt}
	\includegraphics[width=1\linewidth]{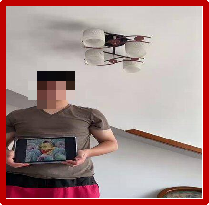}\vspace{1pt}
	\includegraphics[width=1\linewidth]{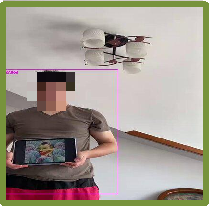}\vspace{1pt}
	\includegraphics[width=1\linewidth]{phyexp-2.jpg}\vspace{1pt}
	\end{minipage}}
	\subfigure[Natural Patch]{
	\begin{minipage}[t]{0.11\linewidth}
	\includegraphics[width=1\linewidth]{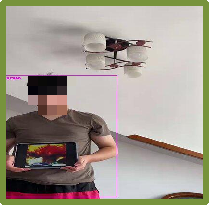}\vspace{1pt}
	\includegraphics[width=1\linewidth]{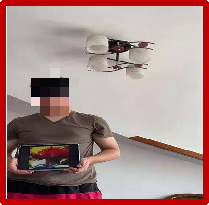}\vspace{1pt}
	\includegraphics[width=1\linewidth]{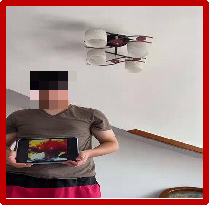}\vspace{1pt}
	\includegraphics[width=1\linewidth]{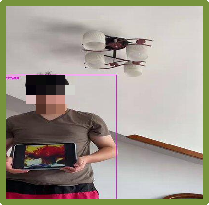}\vspace{1pt}
	\includegraphics[width=1\linewidth]{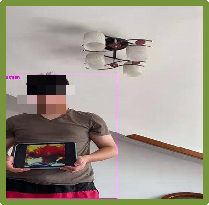}\vspace{1pt}
	\includegraphics[width=1\linewidth]{phyexp-5.png}\vspace{1pt}
	\end{minipage}}
	\subfigure[Natural Patch]{
	\begin{minipage}[t]{0.11\linewidth}
	\includegraphics[width=1\linewidth]{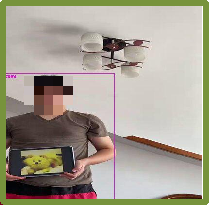}\vspace{1pt}
	\includegraphics[width=1\linewidth]{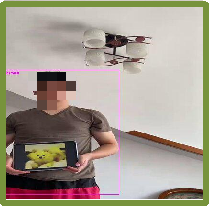}\vspace{1pt}
	\includegraphics[width=1\linewidth]{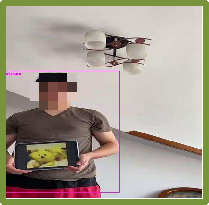}\vspace{1pt}
	\includegraphics[width=1\linewidth]{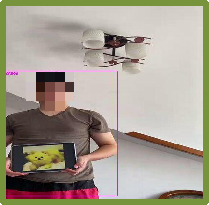}\vspace{1pt}
	\includegraphics[width=1\linewidth]{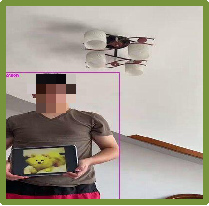}\vspace{1pt}
	\includegraphics[width=1\linewidth]{phyexp-4.jpg}\vspace{1pt}
	\end{minipage}}
	\subfigure[Adv. Patch]{
	\begin{minipage}[t]{0.11\linewidth}
	\includegraphics[width=1\linewidth]{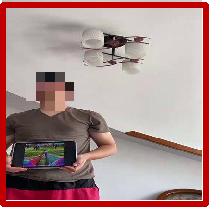}\vspace{1pt}
	\includegraphics[width=1\linewidth]{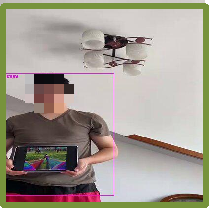}\vspace{1pt}
	\includegraphics[width=1\linewidth]{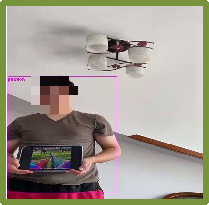}\vspace{1pt}
	\includegraphics[width=1\linewidth]{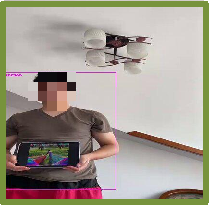}\vspace{1pt}
	\includegraphics[width=1\linewidth]{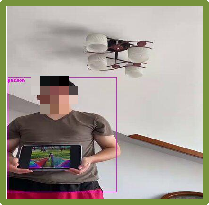}\vspace{1pt}
	\includegraphics[width=1\linewidth]{phyexp-3.png}\vspace{1pt}
	\end{minipage}}
	\subfigure[Adv. Texture]{
	\begin{minipage}[t]{0.11\linewidth}
	\includegraphics[width=1\linewidth]{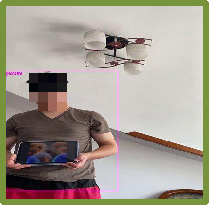}\vspace{1pt}
	\includegraphics[width=1\linewidth]{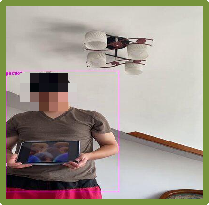}\vspace{1pt}
	\includegraphics[width=1\linewidth]{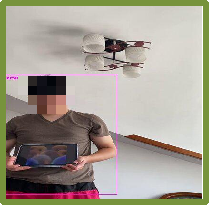}\vspace{1pt}
	\includegraphics[width=1\linewidth]{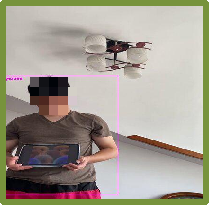}\vspace{1pt}
	\includegraphics[width=1\linewidth]{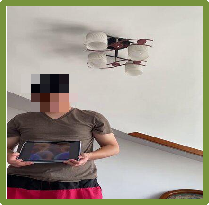}\vspace{1pt}
	\includegraphics[width=1\linewidth]{phyexp-6.png}\vspace{1pt}
	\end{minipage}}
	\subfigure[Benign Image]{
	\begin{minipage}[t]{0.11\linewidth}
	\includegraphics[width=1\linewidth]{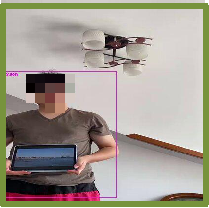}\vspace{1pt}
	\includegraphics[width=1\linewidth]{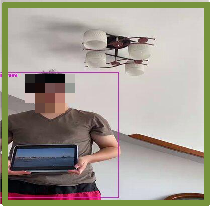}\vspace{1pt}
	\includegraphics[width=1\linewidth]{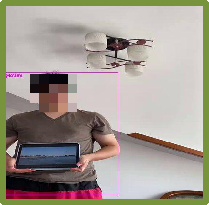}\vspace{1pt}
	\includegraphics[width=1\linewidth]{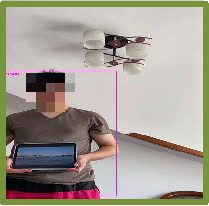}\vspace{1pt}
	\includegraphics[width=1\linewidth]{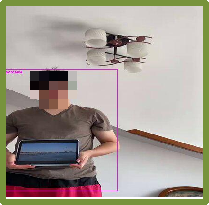}\vspace{1pt}
	\includegraphics[width=0.99 \linewidth]{phyexp-7.jpg}\vspace{1pt}
	\end{minipage}}
	\caption{This illustration demonstrates the experimental outcomes in the physical domain. Images with a green rectangle represent successful recognition of the person's class by the object detectors, whereas images with a red rectangle indicate the object detectors' failure to recognize the person's class. The images at the bottom are adversarial patches and mosaics are added after detection, to avoid privacy disclosure for volunteers.}
	\label{f10}
	\end{figure*}
\subsubsection{Comparision with Other Adversarial Patches}
Table \ref{t6} presents the empirical results of various adversarial patches within the physical domain, including AdvPatch, AdvTexture, a benign image, Natural Patch, and the proposed MVPatch. Volunteers held these patches against their chests at angles of {$0^\circ, 30^\circ, 60^\circ, 90^\circ$} and distances of {1, 2, 3, 4} meters. Using an iPhone 11, we captured a total of 9,681 images from various angles and distances, and reviewed them on an iPad Pro to ensure pixel detail preservation. The table evaluates key metrics such as Attack Success Rate (ASR) on the YOLOv2 model, number of images, stealthiness, physical applicability, transferability, and computational cost. AdvPatch and AdvTexture are classified as meaningless adversarial patches, while Natural Patch and MVPatch are meaningful. The benign image contains no adversarial features. ASR measures the attack success rate, "Number of Images" indicates the experimental samples, "Stealthiness" reflects naturalness, "Physical Attack" confirms real-world applicability, "Transferable Attack" indicates cross-domain effectiveness, and "Lightweight Model" denotes low computational cost. The proposed MVPatch demonstrates superior performance across all evaluated metrics.

\textbf{Analysis and Discussion:} As Table \ref{t6} illustrates, both AdvPatch(Fig.\ref{f6a}) and AdvTetxure(Fig.\ref{f6b}) do not have much more attack performance on the YOLOv2 model in the physical world and the ASRs are 13.57\% and 12.72\%, respectively. The Natural Patch(Fig.\ref{f6e}) has higher ASR and its value of ASR is 19.43\%, but the other Natural Patch(Fig.\ref{f6d}) does not perform as well as Fig.\ref{f6e}. The ASR of the first Natural Patch(Fig.\ref{f6d}) is almost close to the Benign Image(Fig.\ref{f6c}). The MVPatches perform better than other adversarial patches in the physical world. The ASR of MVPatches are respectively 22.60\% (Fig.\ref{f6f}) and 26.33\% (Fig.\ref{f6g}). Consequently, we ensure the MVPatch has more attack performance than others in the physical world.
\subsubsection{Diverse Angles}
We collected 6,846 images from various angles ($0^\circ$, $30^\circ$, $60^\circ$, $90^\circ$) to assess the impact of angle variation on the attack performance of different adversarial patches, including Natural Patch, AdvPatch, AdvTexture, Benign Image, and our MVPatch. To ensure experimental fairness, angles were adjusted to {$0^\circ$, $30^\circ$, $60^\circ$, $90^\circ$} while maintaining a constant distance of 2 meters. The volunteer rotated, with the adversarial patch rotating in the opposite direction to ensure it always faced the camera, a setup referred to as "with Rotation." The experimental results are presented in Fig.\ref{f7a}. The light blue dotted line with squares represents the MVPatch with Rotation; the red dotted line with circles represents the Benign Image with Rotation; the dark blue dotted line with triangles represents the AdvPatch with Rotation; the green dotted line with triangles represents the AdvTexture with Rotation; and the black dotted line with squares represents the Natural Patch with Rotation. The X-axis represents the angles, while the Y-axis represents the attack success rate.

\textbf{Analysis and Discussion:} The effectiveness of adversarial patches with rotation diminishes as the angle of rotation increases. At an angle of $0^\circ$, the ASR of the MVPatch can reach 39.5\%, and other patches perform well at the angle of $0^\circ$ except for the AdvTexture. At the angle of $90^\circ$, the ASRs of the Natural Patch and AdvTexture surpass the MVPatch at the same angle. Even a slight variation in the angle can significantly enhance the robustness and aggressiveness of the adversarial patch when it is rotated. Fig.\ref{f7a} also demonstrates that the MVPatch outperforms other attack methods in terms of angle variations (at the angles of $0^\circ$, $30^\circ$, and $60^\circ$).
\subsubsection{Diverse Distance}	
A total of 3,984 images were collected at various distances (1 meter, 2 meters, 3 meters, and 4 meters) to evaluate the impact of distance on the attack performance of adversarial patches, including MVPatch, Natural Patch, AdvPatch, AdvTexture, and Benign Image. To ensure experimental fairness, distances were varied {1, 2, 3, 4} meters while maintaining a constant angle of $0^\circ$. The results are presented in Fig.\ref{f7b}. The light blue dotted line with squares represents the MVPatch, the black dotted line with squares represents the Natural Patch, the dark blue dotted line with triangles represents the AdvPatch, the green dotted line with triangles represents the AdvTexture, and the red dotted line with circles represents the Benign Image. The X-axis denotes the attack success rate, and the Y-axis indicates the distance from the volunteer with the adversarial patch to the camera.

\textbf{Analysis and Discussion:} 
As Fig.\ref{f7b} shows, the adversarial patches generated by MVPatch and AdvPatch exhibit high levels of aggression when the volunteer is at a distance of one meter from the camera. The ASR of MVPatch is 71.67\% and the ASR of AdvPatch is 63.46\% at a distance of 1 meter. As the volunteer moves farther away from the camera (beyond one meter), the effectiveness of the attack diminishes. However, the adversarial patches generated by AdvTexture and Natural Patch have higher aggression at a distance of 2 meters.  When the volunteer moves farther away from the camera (beyond two meters), the effectiveness of the attack diminishes. Similarly, the attack weakens as the volunteer moves closer.
\subsection{Visualization}
This section provides experimental results in both the digital and physical domains to visualize our proposed approach. The green rectangle represents a successful detection of a person by the object detector, while the red rectangle indicates a failure to detect the person class. In the digital domain, we choose various experimental images from the Inria Person dataset as evaluation objects. In the physical domain, the source images of this experiment are not covered with mosaics. For privacy protection, mosaics are later added near the face. Fig.\ref{f9} displays the results of digital experiments, while Fig.\ref{f10} depicts the results of physical experiments.
\section{Conclusion and Limitations}\label{s6}
In this paper, we propose a Dual-Perception-Based Framework (DPBF) that includes a Model-Perception-Based Module (MPBM), a Human-Perception-Based Module (HPBM), and regularization terms. Our objective is to create vivid and aggressive adversarial patches, termed \textbf{M}ore \textbf{V}ivid \textbf{Patch} (\textbf{MVPatch}), designed to camouflage individuals in detection models for real-world applicability. The MPBM leverages ensemble attacks across various object detectors to enhance the transferability of adversarial patches, demonstrating improved generalization and stability through generalization error analysis. Meanwhile, the HPBM introduces a lightweight approach for visual similarity measurement, which makes the adversarial patches less noticeable and incorporates additional transformations to further enhance transferability while preserving stealthiness. Regularization terms are employed to increase the practicality of the generated patches in physical environments. We also introduce metrics for naturalness and transferability to provide an unbiased assessment of the adversarial patches. Extensive qualitative and quantitative experiments reveal that minimizing our proposed objective function results in the most effective patches, maintaining competitive performance compared to existing methods. The MVPatch demonstrates exceptional stealthiness and transferability across both digital and physical domains, achieving these results with minimal computational cost.\\ 
\textbf{Limitations.} Despite the advancements presented, our approach has several limitations. Factors such as angles and distances can still affect the effectiveness of adversarial attacks. Additionally, the vulnerability of adversarial patches against state-of-the-art object detection models remains inadequately explored. Future work will address these issues by developing methods to minimize the impact of angle and distance, thereby optimizing attack performance in both digital and physical contexts. We also intend to apply our technique to a wider range of off-the-shelf object detection models to evaluate their susceptibility to adversarial patches. Furthermore, we plan to explore more robust defense mechanisms to mitigate the risk of adversarial attacks in both digital and physical environments.




%





\ifCLASSOPTIONcaptionsoff
  \newpage
\fi





\bibliographystyle{IEEEtran}
\bibliography{IEEEabrv,main}

\begin{thebibliography}{10}
\providecommand{\url}[1]{#1}
\csname url@rmstyle\endcsname
\providecommand{\newblock}{\relax}
\providecommand{\bibinfo}[2]{#2}
\providecommand\BIBentrySTDinterwordspacing{\spaceskip=0pt\relax}
\providecommand\BIBentryALTinterwordstretchfactor{4}
\providecommand\BIBentryALTinterwordspacing{\spaceskip=\fontdimen2\font plus
\BIBentryALTinterwordstretchfactor\fontdimen3\font minus \fontdimen4\font\relax}
\providecommand\BIBforeignlanguage[2]{{%
\expandafter\ifx\csname l@#1\endcsname\relax
\typeout{** WARNING: IEEEtran.bst: No hyphenation pattern has been}%
\typeout{** loaded for the language `#1'. Using the pattern for}%
\typeout{** the default language instead.}%
\else
\language=\csname l@#1\endcsname
\fi
#2}}

\bibitem{b96}
Y.-C.-T. Hu, J.-C. Chen, B.-H. Kung, K.-L. Hua, and D.~S. Tan, ``Naturalistic physical adversarial patch for object detectors,'' in \emph{IEEE/CVF International Conference on Computer Vision (ICCV)}, 2021, pp. 7828--7837.

\bibitem{b87}
S.~Thys, W.~Van~Ranst, and T.~Goedemé, ``Fooling automated surveillance cameras: adversarial patches to attack person detection,'' in \emph{IEEE/CVF Conference on Computer Vision and Pattern Recognition Workshop (CVPRW)}, 2019.

\bibitem{b100}
Z.~Hu, S.~Huang, X.~Zhu, X.~Hu, F.~Sun, and B.~Zhang, ``Adversarial texture for fooling person detectors in the physical world,'' \emph{IEEE/CVF Conference on Computer Vision and Pattern Recognition (CVPR)}, pp. 13\,297--13\,306, 2022.

\bibitem{b133}
A.~Krizhevsky, I.~Sutskever, and G.~E. Hinton, ``Imagenet classification with deep convolutional neural networks,'' \emph{Advances in Neural Information Processing Systems (NeurIPS)}, vol.~25, 2012.

\bibitem{b134}
W.~Wang, J.~Shen, and L.~Shao, ``Video salient object detection via fully convolutional networks,'' \emph{IEEE Transactions on Image Processing (TIP)}, vol.~27, no.~1, pp. 38--49, 2017.

\bibitem{b135}
K.~He, X.~Zhang, S.~Ren, and J.~Sun, ``Deep residual learning for image recognition,'' in \emph{IEEE/CVF Conference on Computer Vision and Pattern Recognition (CVPR)}, 2016, pp. 770--778.

\bibitem{b137}
J.~Li, H.~Chang, J.~Yang, W.~Luo, and Y.~Fu, ``Visual representation and classification by learning group sparse deep stacking network,'' \emph{IEEE Transactions on Image Processing (TIP)}, vol.~27, no.~1, pp. 464--476, 2017.

\bibitem{b138}
I.~Sutskever, O.~Vinyals, and Q.~V. Le, ``Sequence to sequence learning with neural networks,'' \emph{Advances in Neural Information Processing Systems (NeurIPS)}, vol.~27, 2014.

\bibitem{b139}
J.~Hirschberg and C.~D. Manning, ``Advances in natural language processing,'' \emph{Science}, vol. 349, no. 6245, pp. 261--266, 2015.

\bibitem{b140}
A.-r. Mohamed, G.~E. Dahl, and G.~Hinton, ``Acoustic modeling using deep belief networks,'' \emph{IEEE Transactions on Audio, Speech, and Language Processing (TASLP)}, vol.~20, no.~1, pp. 14--22, 2011.

\bibitem{b141}
G.~Hinton, L.~Deng, D.~Yu, G.~E. Dahl, A.-r. Mohamed, N.~Jaitly, A.~Senior, V.~Vanhoucke, P.~Nguyen, T.~N. Sainath, \emph{et~al.}, ``Deep neural networks for acoustic modeling in speech recognition: The shared views of four research groups,'' \emph{IEEE Signal Processing Magazine}, vol.~29, no.~6, pp. 82--97, 2012.

\bibitem{b1}
C.~Szegedy, W.~Zaremba, I.~Sutskever, J.~Bruna, D.~Erhan, I.~J. Goodfellow, and R.~Fergus, ``Intriguing properties of neural networks,'' in \emph{International Conference on Learning Representations (ICLR)}, vol. abs/1312.6199, 2014.

\bibitem{b2}
I.~J. Goodfellow, J.~Shlens, and C.~Szegedy, ``Explaining and harnessing adversarial examples,'' in \emph{International Conference on Learning Representations (ICLR)}, vol. abs/1412.6572, 2015.

\bibitem{b136}
Y.~Zhang, X.~Tian, Y.~Li, X.~Wang, and D.~Tao, ``Principal component adversarial example,'' \emph{IEEE Transactions on Image Processing (TIP)}, vol.~29, pp. 4804--4815, 2020.

\bibitem{b161}
Z.~Zhou, J.~Liu, and Y.~Han, ``Adversarial examples are closely relevant to neural network models - a preliminary experiment explore,'' in \emph{Advances in Swarm Intelligence}, Y.~Tan, Y.~Shi, and B.~Niu, Eds.\hskip 1em plus 0.5em minus 0.4em\relax Cham: Springer International Publishing, 2022, pp. 155--166.

\bibitem{b92}
X.~Wei, Y.~Guo, and J.~Yu, ``Adversarial sticker: A stealthy attack method in the physical world,'' \emph{IEEE Transactions on Pattern Analysis and Machine Intelligence (TPAMI)}, 2022.

\bibitem{b119}
X.~Wei, Y.~Guo, J.~Yu, and B.~Zhang, ``Simultaneously optimizing perturbations and positions for black-box adversarial patch attacks,'' \emph{IEEE Transactions on Pattern Analysis and Machine Intelligence (TPAMI)}, vol.~45, no.~7, pp. 9041--9054, 2023.

\bibitem{b122}
Y.~Wang, E.~Sarkar, W.~Li, M.~Maniatakos, and S.~E. Jabari, ``Stop-and-go: Exploring backdoor attacks on deep reinforcement learning-based traffic congestion control systems,'' \emph{IEEE Transactions on Information Forensics and Security (TIFS)}, vol.~16, pp. 4772--4787, 2021.

\bibitem{b123}
X.~Yuan, S.~Hu, W.~Ni, X.~Wang, and A.~Jamalipour, ``Deep reinforcement learning-driven reconfigurable intelligent surface-assisted radio surveillance with a fixed-wing uav,'' \emph{IEEE Transactions on Information Forensics and Security (TIFS)}, vol.~18, pp. 4546--4560, 2023.

\bibitem{b3}
A.~Madry, A.~Makelov, L.~Schmidt, D.~Tsipras, and A.~Vladu, ``Towards deep learning models resistant to adversarial attacks,'' in \emph{International Conference on Learning Representations (ICLR)}, 2018.

\bibitem{b4}
N.~Carlini and D.~A. Wagner, ``Towards evaluating the robustness of neural networks,'' in \emph{IEEE European Symposium on Security and Privacy (SS\&P)}, 2017, pp. 39--57.

\bibitem{b5}
S.-M. Moosavi-Dezfooli, A.~Fawzi, and P.~Frossard, ``Deepfool: A simple and accurate method to fool deep neural networks,'' in \emph{IEEE/CVF Conference on Computer Vision and Pattern Recognition (CVPR)}, 2016, pp. 2574--2582.

\bibitem{b32}
A.~Kurakin, I.~J. Goodfellow, and S.~Bengio, ``Adversarial examples in the physical world,'' in \emph{International Conference on Learning Representations (ICLR)}, Toulon, France, April 2017.

\bibitem{b8}
N.~Papernot, P.~McDaniel, S.~Jha, M.~Fredrikson, Z.~B. Celik, and A.~Swami, ``The limitations of deep learning in adversarial settings,'' in \emph{IEEE European Symposium on Security and Privacy (SS\&P)}, 2016, pp. 372--387.

\bibitem{b44}
P.-Y. Chen, H.~Zhang, Y.~Sharma, J.~Yi, and C.-J. Hsieh, ``Zoo: Zeroth order optimization based black-box attacks to deep neural networks without training substitute models,'' in \emph{ACM Workshop on Artificial Intelligence and Security}, 2017.

\bibitem{b62}
P.~Neekhara, S.~S. Hussain, P.~Pandey, S.~Dubnov, J.~McAuley, and F.~Koushanfar, ``Universal adversarial perturbations for speech recognition systems,'' in \emph{Proceedings of the International Speech Communication Association}, 2019, pp. 481--485.

\bibitem{b113}
J.~Su, D.~V. Vargas, and K.~Sakurai, ``One pixel attack for fooling deep neural networks,'' \emph{IEEE Transactions on Evolutionary Computation (TEC)}, vol.~23, pp. 828--841, 2019.

\bibitem{b114}
A.~Athalye, L.~Engstrom, A.~Ilyas, and K.~Kwok, ``Synthesizing robust adversarial examples,'' in \emph{International Conference on Machine Learning (ICML)}, 2017.

\bibitem{b111}
M.~Sharif, S.~Bhagavatula, L.~Bauer, and M.~K. Reiter, ``Accessorize to a crime: Real and stealthy attacks on state-of-the-art face recognition,'' \emph{ACM SIGSAC Conference on Computer and Communications Security (CCS)}, 2016.

\bibitem{b109}
S.~A. Komkov and A.~Petiushko, ``Advhat: Real-world adversarial attack on arcface face id system,'' in \emph{International Conference on Pattern Recognition (ICPR)}, 2020, pp. 819--826.

\bibitem{b115}
T.~B. Brown, D.~Man{\'e}, A.~Roy, M.~Abadi, and J.~Gilmer, ``Adversarial patch,'' \emph{ArXiv}, vol. abs/1712.09665, 2017.

\bibitem{b93}
X.~Liu, H.~Yang, Z.~Liu, L.~Song, H.~Li, and Y.~Chen, ``Dpatch: An adversarial patch attack on object detectors,'' in \emph{AAAI Workshop on Artificial Intelligence Safety}, Hawaii, USA, Jun 2019.

\bibitem{b99}
Z.~Wu, S.-N. Lim, L.~S. Davis, and T.~Goldstein, ``Making an invisibility cloak: Real world adversarial attacks on object detectors,'' in \emph{European Conference on Computer Vision (ECCV)}.\hskip 1em plus 0.5em minus 0.4em\relax Springer, 2020, pp. 1--17.

\bibitem{b101}
K.~Xu, G.~Zhang, S.~Liu, Q.~Fan, M.~Sun, H.~Chen, P.-Y. Chen, Y.~Wang, and X.~Lin, ``Adversarial t-shirt! evading person detectors in a physical world,'' in \emph{European Conference on Computer Vision (ECCV)}.\hskip 1em plus 0.5em minus 0.4em\relax Springer, 2020, pp. 665--681.

\bibitem{b116}
H.~Huang, Z.~Chen, H.~Chen, Y.~Wang, and K.~A. Zhang, ``T-sea: Transfer-based self-ensemble attack on object detection,'' in \emph{IEEE/CVF Conference on Computer Vision and Pattern Recognition (CVPR)}, Vancouver, Canada, June 2023.

\bibitem{b108}
J.~Wang, A.~Liu, Z.~Yin, S.~Liu, S.~Tang, and X.~Liu, ``Dual attention suppression attack: Generate adversarial camouflage in physical world,'' in \emph{IEEE/CVF Conference on Computer Vision and Pattern Recognition (CVPR)}, 2021, pp. 8561--8570.

\bibitem{b121}
A.~Liu, X.~Liu, J.~Fan, Y.~Ma, A.~Zhang, H.~Xie, and D.~Tao, ``Perceptual-sensitive gan for generating adversarial patches,'' in \emph{Association for the Advancement of Artificial Intelligence (AAAI)}, vol.~33, 07 2019, pp. 1028--1035.

\bibitem{b107}
H.~Xue, A.~Araujo, B.~Hu, and Y.~Chen, ``Diffusion-based adversarial sample generation for improved stealthiness and controllability,'' in \emph{IEEE/CVF Conference on Computer Vision and Pattern Recognition (CVPR)}, Vancouver, Canada, June 2023.

\bibitem{b124}
A.~Guesmi, I.~M. Bilasco, M.~Shafique, and I.~Alouani, ``Advart: Adversarial art for camouflaged object detection attacks,'' \emph{ArXiv}, vol. abs/2303.01734, 2023.

\bibitem{b142}
R.~Girshick, J.~Donahue, T.~Darrell, and J.~Malik, ``Rich feature hierarchies for accurate object detection and semantic segmentation,'' in \emph{IEEE/CVF Conference on Computer Vision and Pattern Recognition (CVPR)}, 2014, pp. 580--587.

\bibitem{b143}
R.~Girshick, ``Fast r-cnn,'' in \emph{IEEE/CVF International Conference on Computer Vision (ICCV)}, 2015, pp. 1440--1448.

\bibitem{b144}
X.~Zhu, S.~Lyu, X.~Wang, and Q.~Zhao, ``Tph-yolov5: Improved yolov5 based on transformer prediction head for object detection on drone-captured scenarios,'' in \emph{IEEE/CVF International Conference on Computer Vision (ICCV)}, 2021, pp. 2778--2788.

\bibitem{b145}
W.~Ren, X.~Wang, J.~Tian, Y.~Tang, and A.~B. Chan, ``Tracking-by-counting: Using network flows on crowd density maps for tracking multiple targets,'' \emph{IEEE Transactions on Image Processing (TIP)}, vol.~30, pp. 1439--1452, 2021.

\bibitem{b146}
M.~Huang, C.~Hou, Q.~Yang, and Z.~Wang, ``Reasoning and tuning: Graph attention network for occluded person re-identification,'' \emph{IEEE Transactions on Image Processing (TIP)}, vol.~32, pp. 1568--1582, 2023.

\bibitem{b148}
Q.~Zhang, R.~Cong, C.~Li, M.-M. Cheng, Y.~Fang, X.~Cao, Y.~Zhao, and S.~Kwong, ``Dense attention fluid network for salient object detection in optical remote sensing images,'' \emph{IEEE Transactions on Image Processing (TIP)}, vol.~30, pp. 1305--1317, 2021.

\bibitem{b130}
J.~Wang, A.~Liu, X.~Bai, and X.~Liu, ``Universal adversarial patch attack for automatic checkout using perceptual and attentional bias,'' \emph{IEEE Transactions on Image Processing (TIP)}, vol.~31, pp. 598--611, 2021.

\bibitem{b131}
Z.~Wei, J.~Chen, M.~Goldblum, Z.~Wu, T.~Goldstein, Y.-G. Jiang, and L.~S. Davis, ``Towards transferable adversarial attacks on image and video transformers,'' \emph{IEEE Transactions on Image Processing (TIP)}, vol.~32, pp. 6346--6358, 2023.

\bibitem{b132}
L.~Huang, C.~Gao, and N.~Liu, ``Erosion attack: Harnessing corruption to improve adversarial examples,'' \emph{IEEE Transactions on Image Processing (TIP)}, 2023.

\bibitem{b164}
P.~Li, Y.~Zhang, L.~Yuan, J.~Zhao, X.~Xu, and X.~Zhang, ``Adversarial attacks on video object segmentation with hard region discovery,'' \emph{IEEE Transactions on Circuits and Systems for Video Technology (TCSVT)}, vol.~34, no.~6, pp. 5049--5062, 2024.

\bibitem{b165}
A.~V. Subramanyam, ``Meta generative attack on person reidentification,'' \emph{IEEE Transactions on Circuits and Systems for Video Technology (TCSVT)}, vol.~33, no.~8, pp. 4429--4434, 2023.

\bibitem{b166}
Z.~Zhou, Y.~Sun, Q.~Sun, C.~Li, and Z.~Ren, ``Only once attack: Fooling the tracker with adversarial template,'' \emph{IEEE Transactions on Circuits and Systems for Video Technology (TCSVT)}, vol.~33, no.~7, pp. 3173--3184, 2023.

\bibitem{b94}
M.~Lee and J.~Z. Kolter, ``On physical adversarial patches for object detection,'' in \emph{International Conference on Machine Learning Workshop (ICMLW)}, Los Angeles, USA, Jun 2019.

\bibitem{b98}
R.~Duan, X.~Ma, Y.~Wang, J.~Bailey, A.~K. Qin, and Y.~Yang, ``Adversarial camouflage: Hiding physical-world attacks with natural styles,'' in \emph{IEEE/CVF Conference on Computer Vision and Pattern Recognition (CVPR)}, 2020, pp. 997--1005.

\bibitem{b149}
Y.~Dong, F.~Liao, T.~Pang, H.~Su, J.~Zhu, X.~Hu, and J.~Li, ``Boosting adversarial attacks with momentum,'' in \emph{IEEE/CVF Conference on Computer Vision and Pattern Recognition (CVPR)}, 2018, pp. 9185--9193.

\bibitem{b150}
J.~Lin, C.~Song, K.~He, L.~Wang, and J.~E. Hopcroft, ``Nesterov accelerated gradient and scale invariance for adversarial attacks,'' in \emph{International Conference on Learning Representations (ICLR)}, 2020.

\bibitem{b151}
C.~Xie, Z.~Zhang, Y.~Zhou, S.~Bai, J.~Wang, Z.~Ren, and A.~L. Yuille, ``Improving transferability of adversarial examples with input diversity,'' in \emph{IEEE/CVF Conference on Computer Vision and Pattern Recognition (CVPR)}, 2019, pp. 2730--2739.

\bibitem{b152}
Y.~Xiong, J.~Lin, M.~Zhang, J.~E. Hopcroft, and K.~He, ``Stochastic variance reduced ensemble adversarial attack for boosting the adversarial transferability,'' in \emph{IEEE/CVF Conference on Computer Vision and Pattern Recognition (CVPR)}, 2022, pp. 14\,983--14\,992.

\bibitem{b153}
A.~Guesmi, R.~Ding, M.~A. Hanif, I.~Alouani, and M.~Shafique, ``Dap: A dynamic adversarial patch for evading person detectors,'' in \emph{IEEE/CVF Conference on Computer Vision and Pattern Recognition (CVPR)}, 2024, pp. 24\,595--24\,604.

\bibitem{b162}
Y.~Ran, W.~Wang, M.~Li, L.-C. Li, Y.-G. Wang, and J.~Li, ``Cross-shaped adversarial patch attack,'' \emph{IEEE Transactions on Circuits and Systems for Video Technology (TCSVT)}, vol.~34, no.~4, pp. 2289--2303, 2024.

\bibitem{b95}
B.~G. Doan, M.~Xue, S.~Ma, E.~Abbasnejad, and D.~C. Ranasinghe, ``Tnt attacks! universal naturalistic adversarial patches against deep neural network systems,'' \emph{IEEE Transactions on Information Forensics and Security (TIFS)}, vol.~17, pp. 3816--3830, 2022.

\bibitem{b110}
L.~Huang, C.~Gao, Y.~Zhou, C.~Xie, A.~Yuille, C.~Zou, and N.~Liu, ``Universal physical camouflage attacks on object detectors,'' \emph{IEEE/CVF Conference on Computer Vision and Pattern Recognition (CVPR)}, pp. 717--726, 2020.

\bibitem{b163}
T.~Wang, L.~Zhu, Z.~Zhang, H.~Zhang, and J.~Han, ``Targeted adversarial attack against deep cross-modal hashing retrieval,'' \emph{IEEE Transactions on Circuits and Systems for Video Technology (TCSVT)}, vol.~33, no.~10, pp. 6159--6172, 2023.

\bibitem{b127}
X.~Qi, K.~Huang, A.~Panda, M.~Wang, and P.~Mittal, ``Visual adversarial examples jailbreak large language models,'' \emph{ArXiv}, vol. abs/2306.13213, 2023.

\bibitem{b128}
D.~Lu, Z.~Wang, T.~Wang, W.~Guan, H.~Gao, and F.~Zheng, ``Set-level guidance attack: Boosting adversarial transferability of vision-language pre-training models,'' \emph{ArXiv}, vol. abs/2307.14061, 2023.

\bibitem{b129}
J.~Zhang, Q.~Yi, and J.~Sang, ``Towards adversarial attack on vision-language pre-training models,'' \emph{ACM International Conference on Multimedia (ACM MM)}, 2022.

\bibitem{b155}
S.~R. Sain, ``The nature of statistical learning theory,'' \emph{Technometrics}, vol.~38, no.~4, pp. 409--409, 1996.

\bibitem{b156}
V.~Vapnik, E.~Levin, and Y.~Le~Cun, ``Measuring the vc-dimension of a learning machine,'' \emph{Neural Computation}, vol.~6, no.~5, pp. 851--876, 1994.

\bibitem{b157}
P.~Bartlett, O.~Bousquet, and S.~Mendelson, ``Local rademacher complexities,'' \emph{The Annals of Statistics}, vol.~33, no.~4, pp. 1497--1537, 2005.

\bibitem{b158}
D.-X. Zhou, ``The covering number in learning theory,'' \emph{Journal of Complexity}, vol.~18, no.~3, pp. 739--767, 2002.

\bibitem{b159}
D.~Stutz, M.~Hein, and B.~Schiele, ``Disentangling adversarial robustness and generalization,'' in \emph{IEEE/CVF Conference on Computer Vision and Pattern Recognition (CVPR)}, 2019, pp. 6976--6987.

\bibitem{b160}
N.~D. Singh, F.~Croce, and M.~Hein, ``Revisiting adversarial training for imagenet: Architectures, training and generalization across threat models,'' \emph{Advances in Neural Information Processing Systems (NeurIPS)}, vol.~36, 2024.

\bibitem{b120}
H.~Ma, K.~Xu, X.~Jiang, Z.~Zhao, and T.~Sun, ``Transferable black-box attack against face recognition with spatial mutable adversarial patch,'' \emph{IEEE Transactions on Information Forensics and Security (TIFS)}, pp. 1--1, 2023.

\bibitem{b118}
Z.~Chen, B.~Li, S.~Wu, S.~Ding, and W.~Zhang, ``Query-efficient decision-based black-box patch attack,'' \emph{IEEE Transactions on Information Forensics and Security (TIFS)}, vol.~18, pp. 5522--5536, 2023.

\bibitem{b154}
M.~Belkin, S.~Ma, and S.~Mandal, ``To understand deep learning we need to understand kernel learning,'' in \emph{International Conference on Machine Learning (ICML)}, 2018, pp. 541--549.

\bibitem{b117}
N.~Dalal and B.~Triggs, ``Histograms of oriented gradients for human detection,'' in \emph{IEEE/CVF Conference on Computer Vision and Pattern Recognition (CVPR)}, vol.~1, 2005, pp. 886--893 vol. 1.

\end{thebibliography}

\vfill


\end{document}